\documentclass[journal,onecolumn,draftclsnofoot,12pt]{IEEEtran}%
 \usepackage[overload]{empheq}
\usepackage{setspace}
\usepackage{lettrine}

\usepackage{graphics,
           psfrag,
           epsfig,
           amsthm,
           cite,
           amssymb,
           url,
           dsfont,
           subfigure,
           algorithmic,
           balance,
           enumerate,
           color,
           setspace
}

\usepackage{capt-of}
\usepackage{booktabs}
\usepackage{varwidth}

\usepackage{amsmath}
\usepackage{epstopdf}

\newtheorem{example}{Example}
\newtheorem{definition}{Definition}

\newtheorem{lemma}{Lemma}

\newtheorem{theorem}{Theorem}
\newtheorem{corollary}{Corollary}

\usepackage{ragged2e}
\newcommand{\eref}[1]{(\ref{#1})}
\newcommand{\sref}[1]{Section~\ref{#1}}

\newcommand{\pref}[1]{Proposition~\ref{#1}}
\newcommand{\cref}[1]{Constraint~\ref{#1}}
\newcommand{\thref}[1]{Theorem~\ref{#1}}
\newcommand{\corref}[1]{Corollary~\ref{#1}}

\usepackage[letterpaper, left=1in, right=1in, bottom=1in, top=1in]{geometry}
\newcommand{\algref}[1]{Algorithm~\ref{#1}}

\newcommand{\ignore}[1]{}

\IEEEoverridecommandlockouts
\usepackage{mathtools}
\usepackage{color}
\usepackage[utf8]{inputenc}
\usepackage[ruled, lined, longend]{algorithm2e}
\SetEndCharOfAlgoLine{}
\usepackage{color}
\begin{document}

\title{\vspace{-.35cm} 
Completion Time Minimization in Fog-RANs using D2D Communications and Rate-Aware Network Coding}

\author{
 \IEEEauthorblockN{Mohammed S. Al-Abiad, \textit{Student Member, IEEE}, and Md. Jahangir Hossain, \textit{Senior Member, IEEE} \vspace{-1.6cm}}
\thanks {
Mohammed S. Al-Abiad and Md. Jahangir Hossain are with the School of Engineering, University of British Columbia, Kelowna, BC V1V 1V7, Canada (e-mail: m.saif@alumni.ubc.ca, jahangir.hossain@ubc.ca).
}
}

\maketitle
\begin{abstract}
\ignore{Device-to-device (D2D) communications and fog radio access
network (F-RAN) have been introduced to meet growing users' demand.} The  device-to-device communication-aided fog radio access network, referred to as \textit{D2D-aided F-RAN}, takes advantage of caching at enhanced remote radio heads
(eRRHs) and D2D proximity for improved system performance. For D2D-aided F-RAN,  we develop a framework that exploits the cached contents at
eRRHs, their transmission rates/powers, and previously received contents by different users to deliver the
requesting contents to users with a minimum completion time. Given the
intractability of the completion time minimization problem,  we formulate  it at each transmission by approximating the completion time and  decoupling it into two subproblems. In the
first subproblem, we minimize the possible completion time in
eRRH downlink transmissions, while in the second subproblem, we maximize the number of users to be scheduled on D2D links. We design two theoretical graphs, namely \textit{interference-aware instantly decodable network coding (IA-IDNC)} and \textit{D2D conflict}  graphs to reformulate two subproblems as maximum weight clique and maximum independent set problems, respectively. Using these graphs, we heuristically develop joint and coordinated scheduling  approaches. Through extensive  simulation results, we demonstrate the effectiveness of the proposed schemes against existing baseline schemes. Simulation results show that the proposed two approaches achieve a considerable
performance gain in terms of the completion time minimization. \vspace{-0.5cm}
\end{abstract}

\begin{IEEEkeywords}
\vspace{-0.4cm}
Device-to-device communications, fog radio access networks, coordinated scheduling, network coding, power allocation, time-critical applications.
\vspace{-0.5cm}
\end{IEEEkeywords}

\section{Introduction}
\vspace{-0.28cm}
\IEEEPARstart{W}{ith} the explosive increase in users' demand, the data rate and Quality-of-Service (QoS) performance of
current radio access networks need to be improved significantly
\cite{1n}. Cloud radio access
network (C-RAN) is a promising solution to improve the QoS for its users and support the exponentially growing demands \cite{2n}. The cloud
base-station (CBS) in C-RAN connects with distributed remote radio heads (RRHs) via fornthual links for  cooperative transmission \cite{3n,4n, 7nnnn, 27n}. Since the requested contents are not
cached at the RRHs, the capacity and
delay constrained  of fronthaul links limit the performance of C-RANs to
meet the growing demand in 5G cellular networks \cite{5n}. Therefore,
Fog-RAN (F-RAN) has been introduced that exploits both edge caching and C-RAN for carrying out content delivery effectively \cite{6nn}. In F-RAN, the so called \textit{enhanced RRHs} (eRRHs) support high caching capability. 

In order to further improve the performance of F-RANs, implementing device-to-device (D2D) communications \cite{6nnnn} in F-RAN is shown to be a potential technology in 5G and beyond. This integrated system is referred as D2D-aided F-RAN \cite{9nn}. D2D-aided F-RAN system draws a remarkable benefit for reducing both users' contents delivery time and burden on fronthaul links.  Thanks to the edge caching at the eRRHs and users' cooperation via D2D communications, this paper is focused on content delivery problem in D2D-aided F-RAN system. The content delivery problem of interest is motivated by immediate delivery of common popular contents for real-time applications, i.e., live video streaming. In particular, we study the scheduling of content delivery problem from both the eRRHs and potential transmitting users in D2D-aided F-RAN system using network coding (NC) \cite{6nnn}.

The problem of delivering contents, i.e., a frame of delay-sensitive files, to a set of users  with minimum possible delay  has been a topic of research for a quite some time. This problem is referred as \textit{completion time minimization} problem. Based on layer functionalities, existing NC solutions for this problem can be classified  into upper-layer NC  \cite{12n, 13, 13n,14n,15n,16n, 9n, 9nn} and rate aware NC \cite{18n, 19n,19nn, 20n,20nn, 21nn, 21n,22n} methods. As their names indicate, upper layer NC algorithms focused only on NC at the network layer to minimize the number of transmissions. Rate aware NC approaches incorporate both upper and physical layers to minimize the completion time (in second) required to deliver requested files to all requesting users. The latter is more practically relevant as it involves the dynamic nature of wireless channels in the completion time optimization. 

\ignore{Our work considers the downlink of D2D-enabled F-RANs consisting of
several single-antenna users and several single-antenna eRRHs connected to one CBS. These users are connected partially to each other and interested in downloading a set of files. As such, the completion time is minimized. The problem of interest is motivated by real-time applications, i.e.,  video streaming, in which each decoded file
is immediately used at the application layer and partial decoding of contents is crucial to user's
experience. Therefore, we adopt a real-time subclass of NC, namely Instantly Decodable Network Coding (IDNC), at both eRRHs and potential transmitting users. IDNC allows progressive and XOR encoding of source files at the transmitter and  progressive decoding of the received files at the receiver. It requires a small
coefficient reporting overhead and easy for implementation on battery-powered devices  \cite{16n}.}

\subsection{Related Works and Challenges}
\textit{Related Works in Physical Layer:} Most relevant works on C-RANs focused on scheduling users to
RRHs in order to maximize sum-rate, e.g., \cite{7nnnn, 27n}, \cite{33,24n,25n}. The study in \cite{7nnnn} was extended in \cite{27n} to include power allocation optimization for the radio resource blocks. However, these studies are agnostic to the
available side information at network layer, i.e., requested and previously received contents by different users. As a result, each eRRH sends uncoded file that serves a single
user. The term ``uncoded file" is referred to file without NC. It has been observed that users tend to have a common interest in requesting same contents, especially
popular videos, within a small interval of time \cite{16n}. This happens
frequently in a hotspot, e.g., a playground, a
public transport, a conference hall, and so on. In aforementioned schemes, the contents are transmitted  without NC, which
degrade the system performance. Therefore, a  subclass of NC,
namely the Instantly Decodable NC (IDNC), can be exploited to efficiently select a combination of contents (binary XOR combination) that can benefit a subset of interested users.

\textit{Related Works in Network Layer:} The completion time minimization problem in IDNC-based networks  was considered in different network settings, e.g., point-to-multipoint (PMP) \cite{9n, 12n}, D2D networks \cite{13n, 14n}, D2D F-RANs \cite{15n}. In particular, the authors of \cite{9n,12n,13n} proposed schemes to deliver the requested files by users  with a minimum possible number of transmissions. Recently, in \cite{15n}, a centralized D2D F-RAN scheme was proposed for completion time reduction. However, the aforementioned works considered IDNC from the perspective of network-layer. The main drawback is that  the transmission rate of each radio resource block is selected based on the user with the weakest channel quality. This results in prolonged file reception time and thus, consumes the time resources of network. Therefore, considering both network layer coding and physical layer factors, such as transmission rate, is crucial, which is known as \textit{rate-aware IDNC} (RA-IDNC) \cite{18n}.

\textit{Related Works in RA-IDNC}: With RA-IDNC, the completion time minimization problem needs a careful optimization of selecting the IDNC file and transmission rate of each radio resource, see for example \cite{19n, 19nn, 20n, 21n, 22n}. The authors of \cite{21n} used RA-IDNC in C-RANs for completion time reduction. However, the authors assumed that all RRHs maintain a fixed transmit power level. Moreover, for synchronization purposes,  the same transmission rate (i.e., the lowest transmission rate of all RRHs) is selected. This
may violate the QoS rate guarantee and lead to a longer time for file transmission.  Importantly, the proposed solution did not exploit the high capabilities of D2D communications.  Recently, a cross-layer IDNC scheme was proposed for cloud offloading in F-RAN \cite{23nn}. Inspired by \cite{23nn}, our work  addresses the completion time minimization problem in D2D-aided F-RAN system using RA-IDNC and D2D communications.

\textit{Challenges:} The completion time minimization problem in D2D-aided F-RAN involves
many factors, such as power levels of eRRHs, their cached files, users' limited coverage zones, their requested and previously received files, and their heterogeneous
physical-layer capacities. Since all these combinatorial factors need to be jointly considered, such  problem is intractable. Indeed, considering only power  levels factor for solving a fixed schedule (without NC) problem is non-convex \cite{27n}, \cite{33,24n,25n,26n}. To the best of the authors' knowledge, this work is the first attempt to solve the completion time minimization problem in D2D-aided F-RAN system while considering all above factors. Besides the intractability of the aforementioned factors, a key challenge to the problem is the use of IDNC codes in D2D-aided F-RAN as the objective of both techniques can be contradicting.\ignore{   If we greedily combine users' files using IDNC to each eRRH, the number of users left to
be served by users' cooperation over potential D2D links can be reduced. However, the transmission rate of each eRRH should match the lowest capacities among all
its assigned users. This results in a longer time for file transmission which conflicts with the potential of F-RAN. On the other hand, pre-setting a minimum target
transmission rate in each eRRH usually results in assigning a few users to it. This contradicts with the goal of IDNC that aims to combine files that satisfies a significant set of users.  Furthermore, serving few users by eRRHs means a high number of transmitting users over D2D links is needed to serve a set
of (or possibly all) remaining users. Because of half-duplex channels, the transmitting users  would not benefit from D2D transmissions as they transmit. Hence, their completion times are increasing, instated of decreasing.} Therefore, a balance among the conflicting effects of IDNC codes, scheduled users, and transmission rates/powers of eRRHs and users in D2D-aided F-RAN system is crucial
to minimize the
total frame delivery time.

\subsection{Contributions}
In this work, we tackle the completion time minimization problem  in downlink D2D-aided F-RAN settings. To this end, we introduce a novel optimization framework taking network coding, rate/power optimization, potential D2D communications, and users' limited coverage zones into account. In the proposed framework, network-coded transmissions from both eRRHs and potential users are developed to deliver all files to all requesting users in the least amount of time.  
The main contributions of our work are summarized as follows. 
\begin{enumerate}
\item For a D2D-aided F-RAN, we develop a framework where eRRHs and users collaborate  to minimize the completion time. In particular, given the intractability of solving the completion time minimization problem over all possible future NC decisions, we reformulate the problem at each transmission with
the constraints on user scheduling, their limited coverage zones, transmission rates, maximum power
allocations, and QoS rate guarantee. By analyzing the
problem, we decompose it into two subproblems. 
\item The first subproblem aims to obtain the possible completion time in eRRHs transmissions through minimizing the transmission time. To solve
it, we design an Interference-Aware IDNC (IA-IDNC) graph that efficiently solves the user scheduling and power allocation problem jointly under the completion time constraints. Based on this, the transmission time achieved by eRRHs is revealed for solving the second subproblem. Then, we introduce a new D2D conflict graph to heuristically solve the second subproblem, i.e., maximizing the number of
users that can be scheduled on D2D links. The aforementioned graph-based solutions of the corresponding subproblems will be referred to as \textit{Joint Approach}.
\item Since
the IA-IDNC graph in the joint approach grows fast with the NC combinations in large network size, we propose an alternative and efficient low-complexity \textit{coordinated scheduling approach} that solves the completion time problem using graph theoretic method.
\item We compare our proposed schemes with existing coded and uncoded (without NC) schemes. Selected numerical results demonstrate
that the proposed schemes can effectively improve completion time performance.
\end{enumerate}

The rest of this paper is organized as follows. In \sref{SMMM}, we present an overview of the D2D-aided F-RAN system. In \sref{PF}, we describe the NC model and analyze the transmission time for simplifying the expression of the completion time. The completion
time minimization problem at each transmission is formulated and decomposed  in \sref{PFPD}.\ignore{ The designed graphs and problems transformation can be found in \sref{G}.} We solve the problem jointly in \sref{JA} and propose a relative low complexity approach in \sref{CS}. Numerical results are presented in \sref{N}. Finally, \sref{C} concludes the paper.

\subsection{Notations}
Matrices are shown by bold characters, e.g., $\textbf{C}$. Calligraphic
letters denote sets and their corresponding capital letters denote the cardinalities
of these sets, e.g., $N=|\mathcal N|$. Further, $\mathcal P(\mathcal N)$ shows the power set of set $\mathcal N$ and $\mathcal A \times \mathcal B$ shows the Cartesian product
of the two sets $\mathcal A$ and $\mathcal B$.
\section{System Model} \label{SMMM}

\subsection{System Overview}
We consider a D2D-aided F-RAN system, shown in Fig. \ref{fig1}, that
consists of one cloud base station (CBS), $K$ single antenna enhanced remote radio heads (eRRHs), and $N$ users. The sets of eRRHs and
users are denoted by  $\mathcal{K}=\{e_1,e_2,\cdots,e_K\}$, $\mathcal{N}=\{u_1,u_2, \cdots,u_N\}$, respectively. The CBS is responsible for making the NC decisions, power allocation, delivering the instructions to eRRHs
and transmitting users for executions. It also communicates with eRRHs through fronthaul links.  Since
users are allowed to transmit at a certain amount of power, each device has limited coverage zone, denoted by $\mathcal Z_{u_i}$, which represents the service area of the $u_i$-th user to transmit data within a circle of radius $\mathtt R$. Note that user and device are used interchangeably throughout this paper. The set of devices  within the transmission range of the $u_i$-th device is defined by $\mathcal Z_{u_i}=\{u_j\in \mathcal{N}| d^{d2d}_{u_i,u_j}\leq \mathtt R$\}, where $d^{d2d}_{u_i,u_j}$ is the distance between the $u_i$-th and  $u_j$-th devices. Devices can use the same frequency band and transmit encoded files simultaneously via D2D links. We assume there is a set of $F$ popular files, denoted by  $\mathcal{F} =\{f_1,f_2,\ \cdots, f_F\}$. This data frame
constitutes the set of most frequent requested files by the users within a given time duration in a hotspot area. Following the caching model in \cite{23nn}, the $e_n$-th eRRH caches a subset $\mathcal{C}_{e_n}$ that represents its \textit{cache}, i.e., $|\mathcal{C}_{e_n}|=\mu F, \forall e_n\in \mathcal{K}$, where $0\leq\mu\leq 1$ is
the fractional cache size. Further, we assume that all eRRHs collectively cache all files in the frame, i.e., $\bigcup^{K}_{{i=1}}\mathcal{C}_{e_n}=\mathcal {F}$. The distribution of files among eRRHs is assumed to be given, and some common files can be cached in different eRRHs' caches.

In this paper, each device is assumed to be equipped with single antenna and used half-duplex channel. Thus, each device can access to either a D2D channel or cellular channel, and accordingly, it can either transmit or receive at a given time instant. Moreover, the allocated channels for D2D communications  are
assumed to be orthogonal (out-of-band) to those used by eRRHs, i.e., an overlay D2D communication
model is adopted \cite{9nn}. 

\begin{figure}[t!]
\centering
\includegraphics[width=0.45\linewidth]{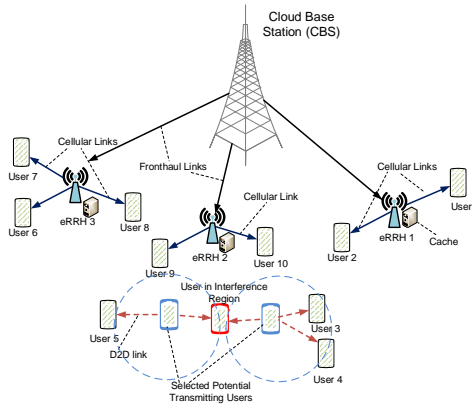}
\caption{Illustration of the D2D-aided F-RAN model with $13$ users, $3$ eRRHs and $1$ CBS.}
\label{fig1}
\end{figure}
\subsection{Physical Layer Model}
The achievable
rate at the $u_i$-th user when receives file
from the $e_n$-th eRRH is given by $
R^{c}_{e_n,u_i}= \log_{2}(1+\text{SINR}_{e_n,u_i}(\textbf{P})), \forall e_n\in \mathcal{K}, \forall u_i\in \mathcal{N},$ where SINR$_{e_n,u_i}(\textbf{P})$ is the corresponding signal-to-interference plus noise-ratio experienced by the $u_i$-th user when it is assigned to the $e_n$-th eRRH. This SINR is given by
\begin{equation}
\text{SINR}_{e_n,u_i}(\textbf{P}) = \cfrac{P_{e_n} |h^{c}_{e_n,u_i}|^{2}}{N_{0} +\sum_{e_{n^\prime}\in \mathcal K, e_{n^\prime}\neq e_n}P_{e_{n^\prime}}| h^{c}_{e_{n^\prime},u_i}|^{2}}, 
\end{equation}
where $h^{c}_{e_n,u_i}$ denotes the channel gain  between the $u_i$-th user and $e_n$-th eRRH, $N_0$ denotes the noise power, $P_{e_n}$ denotes the transmit power of $e_n$-th eRRH, and $\textbf{P} = [P_{e_n}], \forall e_n\in \mathcal K$ is a row vector containing the power levels of the eRRHs in the considered network. The set of users' rates across all eRRHs can be written as 
$\mathcal{R} = \bigotimes_{(e_n,u_i) \in \ \mathcal{\mathcal{K}} \times \mathcal{N}} R^{c}_{e_n,u_i}$,
where the symbol $\bigotimes$ represents the product of the set of the achievable rates.

Similarly, let $h^{d2d}_{u_k,u_i}$ denote the channel gain for the D2D link between the $u_k$-th and $u_i$-th users and $Q_{u_k}$ denote the transmit power of the $u_k$-th user. Then, the achievable rate of D2D pair ($u_k, u_i$) is given by 
$r^{d2d}_{u_k,u_i} = \log_2\left(1+ \cfrac{Q_{u_k} |h^{d2d}_{u_k,u_i}|^{2}}{N_{0} +\sum_{u_{k^\prime}\in \mathcal N_\text{tra}, u_{k^\prime}\neq u_k}Q_{u_{k^\prime}}| h^{d2d}_{u_{k^\prime},u_i}|^{2}}\right), \forall u_k,u_{k^\prime}\in \mathcal N_\text{tra},~ \text{and}~ u_i \in \mathcal{Z}_{u_k}\cap \mathcal Z_{u_{k^\prime}}$, where $\mathcal N_\text{tra}$ is the set of transmitting
users via D2D links.

We assume $h^{c}_{e_n,u_i}$ and $h^{d2d}_{u_k,u_i}$ to be fixed  during a single eRRH and D2D transmissions but change independently from one file transmission to another file transmission.

The channel  capacities of all pairs of D2D links  can be stored  in an $N \times N$  \emph{capacity status matrix (CSM)} $\textbf{r} = [r_{u_k,u_i}],$ $\; \forall (u_k, u_i)$. Since $u_k$-th user does not transmit to itself and cannot transmit to other users outside its coverage zone, $r^{d2d}_{u_k,u_k} = 0$ and $r^{d2d}_{u_k,u_l} = 0, \forall u_l\notin \mathcal Z_{u_k}$.

\section{Network Coding and Completion Time Minimization}\label{PF}
\ignore{In this section, we first discuss the network coding model, and then we analyze the transmission time for simplifying the expression of the completion time. A simple example to understand the completion time metric being optimized  is provided.}

\ignore{\begin{example} \label{ex:capacities}
An example of CSM of D2D links with $D = 6$ devices  is
\begin{equation} \label{eqn:rsm}
\mathbf{R} = \begin{pmatrix}
  0 &  5.3  & 6.5& 3& 5& 2.5\\
  6.1 &  0 & 5.9& 5& 3.5& 6\\
  8.5 &  7.7 & 0& 5.1& 3.4& 2.2\\
  5.5 &  7.2 & 4& 0& 3& 3.4\\
  8.5 &  6.6 & 3.5& 5& 0& 2.2\\
  8 &  3.5 & 2.3& 4.5& 3.4& 0\\
\end{pmatrix}.
\end{equation}
$r_{1,2} = 5.3$ denotes that the channel  capacity from transmitting device $1$ to device $2$ is $5.3$. Moreover,  due to the difference in the transmit powers and thus, the different levels of interference experienced  by each of the devices, this CSM is not symmetric. 
\end{example}}

\subsection{Network Coding in the Network-Layer}
We assume that users are interested in receiving the whole frame $\mathcal F$, and they have already acquired some files in $\mathcal F$ from prior broadcast  transmissions (i.e., without NC) \cite{16n}. The previously acquired files by $u_i$-th user is denoted by the \textit{Has} set $\mathcal H_{u_i}$, and its requested files is denoted by the \textit{Wants} set, i.e., $\mathcal{W}_{u_i} = \mathcal{F}\setminus \mathcal{H}_{u_i}$. Taking advantage of the acquired and requested files
by different users, each eRRH and D2D transmitter can perform
XOR operation on these files and send the combined XORed files to the interested users. As such, the
requested files are delivered to requesting users with minimum completion time.\ignore{ We use transmission/time index $t$ to represent the starting
time of the $t$-th time slot, i.e., $t=1$ refers to the beginning of the first transmission slot.} We use the subscript $t$ to represent the index of transmission/time slot, e.g., $t=1$ refers to the first transmission slot. After each transmission,
each user feedbacks to the eRRHs and neighboring users an acknowledgment for each received file, and accordingly, the \emph{Has}  and \emph{Wants} sets are updated by the CBS \cite{15n}, \cite{16n}. The set of users having \emph{non-empty Wants sets} at the $t$-th transmission slot is denoted by $\mathcal N_{w,t}$, which is defined as $\mathcal N_{w,t} = \{u_i \in \mathcal N | \mathcal{W}_{{u_i},t} \neq \varnothing\}$. When a user receives its requested files, it can act as a D2D transmitter to provide its received files to the interested neighboring users.

Let $\mathtt {f}^{c}_{e_n,t}$ and $\mathtt {f}^{d2d}_{u_k,t}$ denote the XOR file combinations to be sent by the $e_n$-th eRRH and $u_k$-th D2D transmitter, respectively, to the sets of scheduled users $\mathtt u(\mathtt {f}^{c}_{{e_n,t}})$ and $\mathtt u(\mathtt {f}^{d2d}_{{u_k},t})$ at the $t$-th transmission. For simplicity, the subscript transmission index $t$ is often omitted when it is clear from the context. These file combinations $\mathtt {f}^{c}_{e_n}$ and $\mathtt {f}^{d2d}_{u_k}$ are elements of the power sets $\mathcal P({\mathcal C_{e_n}})$ and $\mathcal P({\mathcal H_{u_k}})$, respectively. At every transmission slot $t$, each scheduled user  in $\mathtt u(\mathtt {f}^{c}_{{e_n}})$ can re-XOR $\mathtt {f}^{c}_{{e_n}}$  with its
previously received files to decode a new file. To ensure successful reception at the users, the  maximum transmission rate of a particular transmitting eRRH/user  is equal to the minimum achievable capacity of its scheduled
users. For discussion convenience, the term ``targeted users" is referred to a set of scheduled users who receives an instantly-decodable transmission. Therefore,  the set of targeted users by $e_n$-th eRRH is expressed as $
\mathtt {u}(\mathtt {f}^{c}_{e_n})=\left\{u_i \in \mathcal{N}_w \ \big||\mathtt {f}^{c}_{e_n} \cap \mathcal{W}_{u_i}| = 1~\text{and}~ R^c_{e_n} \leq R^c_{e_n,u_i} \right\}$\ignore{\footnote{We focuse only in this work on  rate adaptation aware IDNC for D2D-aided F-RAN. The transmission loss due to channel impairements is not considered and can be pursued in a future work.}}. 
Similarly, for D2D transmissions, the set of targeted users by $u_k$-th D2D transmitter is expressed as $\mathtt u(\mathtt {f}^{d2d}_{u_k})=\{u_j\in \mathcal N_w\big||\mathtt {f}^{d2d}_{u_k} \cap \mathcal{W}_{u_j}| = 1~\text{and}$ $ u_j\in \mathcal{Z}_{u_k} ~\text{and}~ r^{d2d}_{u_k} \leq r^{d2d}_{u_k,u_j}\}$.
Without loss of generality, the set of all targeted users when $|\mathcal N_\text{tra}|$\ignore{\footnote{The symbol $|\mathcal X|$ represents the cardinality of the set $\mathcal X$.}} D2D transmitters transmit the set of combinations $\mathtt {f}^{d2d}(\mathcal N_\text{tra})$ is represented by $
\mathtt {u}(\mathtt {f}^{d2d}(\mathcal N_{\text{tra}}))$, where $u_k$, $\mathtt {f}^{d2d}_{u_k}$, $\mathtt {u}({\mathtt f}^{d2d}_{u_k})$ are elements in $\mathcal N_\text{tra}$, $\mathtt {f}^{d2d}(\mathcal N_\text{tra})$, and $\mathtt {u}(\mathtt {f}^{d2d}(\mathcal N_\text{tra}))$, respectively.

\ignore{\begin{figure}[t]
\centering
\includegraphics[width=0.4\linewidth]{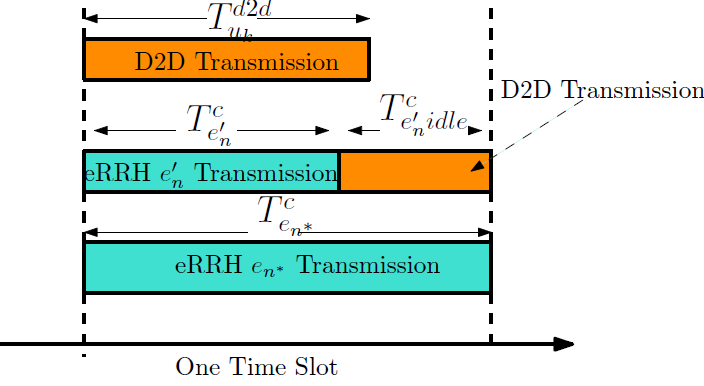}
\caption{Transmission time structure for  eRRHs and potential D2D transmitters for one time slot.}
\label{fig2}
\end{figure}}

\ignore{Hence, the set of potential transmitting users at $t$-th transmission is $\mathcal{A} \subset \mathcal{N}\setminus \tau_b$.  Let $\kappa_{k}$ denote the file combination to be sent by user $k\in \mathcal{A}$. Therefore, user $j\in \tau_{b'}, \forall b'\neq b$ is targeted by the transmission from the $k$-th user if and only if it can guarantee the successful download of a file  from $\mathcal {W}_j$ within the waiting time of any eRRH. Further, this downloaded file must be in the \textit{Has} set of user $k$. Then such file is  \emph{instantly decodable} for user $j$. As mentioned before, the same user cannot be targeted by both transmitting sets $\mathcal{K}$ and $\mathcal{A}$ at a time instant, i.e., $\tau_{\text{b}}\cap \tau{(\kappa(\mathcal{A}))}=\phi$. }

\subsection{Transmission Time Analysis and Expression of the Completion Time}

This subsection provides an analysis of the transmission time for sending coded files from the eRRHs and D2D transmitters to a set of scheduled users, which leads to an  expression of the completion time in D2D-aided F-RAN.
 
The transmission time for sending the coded file $\mathtt {f}^c_{e_n}$ from the $e_n$-th eRRH with rate $R^c_{e_n}$ to the set of targeted users $\mathtt u(\mathtt f^c_{e_n})$ is $T^c_{e_n}=\frac{B}{R_{e_n}}$ seconds, where $B$ is the size of the file in bits. Without loss of generality, let us assume that the $e_{n^{\ast}}$-th eRRH has the minimum rate at the $t$-th transmission slot that is denoted by $R_{e_{n^\ast}}$.  The corresponding transmission duration is $T^c_{e_{n^\ast}}=\frac{B}{R_{e_{n^\ast}}}$ seconds. Since different eRRHs will have different transmission rates, they will have different transmission durations. Thus, the portion of the time that not being used by  $e_{n^\prime}$-th eRRH at $t$-th transmission slot is referred to as the idle time of the $e_{n^\prime}$-th eRRH and denoted by  $T^c_{e_{n^\prime}\text{idle}}$. This idle time can be expressed as $T^c_{e_{n^\prime}\text{idle}}=(T^c_{e_{n^\ast}}-T^c_{e_{n^\prime}})$ seconds. Such idle time can be exploited by the scheduled users of $e_{n^\prime}$-th eRRH via D2D links if it ensures the complete delivery of files, i.e., $T^c_{e_{n^\prime}\text{idle}}\geq T^\text{d2d}_{u_m}$, where $T^\text{d2d}_{u_m}=\frac{B}{r^{d2d}_{u_k}}$ is the transmission duration for sending $\mathtt {f}^{d2d}_{u_m}$ from the $u_m$-th D2D transmitter with adopted rate $r^{d2d}_m$, $\forall u_m\in \mathtt u(\mathtt f^c_{e_{n'}})$. The unscheduled users by the eRRHs can also use D2D links to transmit files, and accordingly, the transmission duration for sending  $\mathtt {f}^{d2d}_{u_k}$ from the $u_k$-th D2D transmitter with adopted rate $r^{d2d}_k$ is $T^{d2d}_{u_k}=\frac{B}{r^{d2d}_{u_k}}$ seconds, $\forall u_k \notin \mathtt u(\mathtt f^c_{e_n}), \forall e_n\in \mathcal K$.
Based on the above discussion, $u_l$-th user experiences one of three possible delays at each transmission, as shown in Fig. \ref{fig2}, and described below.
\begin{enumerate}
\item The time delay for $u_l$-th user receiving a non-instantly decodable transmission  from $e_{n^\ast}$-th eRRH, this delay is $T^c_{e_{n^\ast},u_l}$, $\forall u_l \notin  \mathtt u(\mathtt f^c_{e_{n^\ast}})$.
\item The time delay for $u_l$-th user receiving a non-instantly decodable transmission from $e_{n^\prime}$-th eRRH, this delay is $T^c_{e_{n^\prime},u_l}$,  $e_{n^\prime} \in \mathcal{K}$ and $u_l \notin \mathtt u(\mathtt f^c_{e_{n^\prime}})$.
\item The time delay for $u_l$-th user being transmitting or receiving a non-instantly decodable transmission from any D2D transmitter in the set $\mathcal N_\text{tra}$, this delay is denoted as $T^{d2d}_{u_k,u_l}$, where  $(u_l=u_k) \in \mathcal N_\text{tra}~\text{or}~(u_l \notin \mathtt {u}(\mathtt {f}^{d2d}(\mathcal N_\text{tra}))~\text{and}~u_k \in \mathcal N_\text{tra})$.
\end{enumerate}
Note that for $u_l$-th user, $T^c_{e_{n^\prime},u_l}$ is less than $T^c_{e_{n^\ast},u_l}$ and $\max\limits_{{u_k}\in\mathcal N_\text{tra}} (T^{d2d}_{u_k,u_l})$. Thus, the maximum delay experienced by $u_l$-th user, which is not scheduled at the $t$-th
transmission slot, is equal to $T_{\max,t}=\max(T^c_{e_{n^\ast}}, \max\limits_{{u_k}\in\mathcal N_\text{tra}} (T^{d2d}_{u_k}))$. Consequently,
users that are not scheduled at transmission slot $t$, 
experience $T_{\max,t}$ seconds of delay in a cumulative manner defined as follows.

\begin{definition} A user with non-empty Wants set experiences $T_{\max,t}$ seconds of time delay if it does not receive any requested file at  $t$-th transmission slot. The accumulated time delay of  $u_l$-th user is the sum of $T_{\max,t}$ seconds at each transmission until $t$-th transmission slot, and denoted by $\mathbb{D}_{u_l,t}$. It can be expressed as \begin{align} \label{eq3}
\mathbb{D}_{u_l,t} = \mathbb{D}_{u_l,t-1}+
\begin{cases}
T_{\max,t}  & \text{if} ~u_l \notin  \left(\mathtt u(\mathtt f^{c}_{e_{n^\ast},t})\cup\mathtt u(\mathtt f^{c}_{e_{n^\prime},t})\right), \forall e_{n^\ast}, e_{n^\prime} \in \mathcal K\\
T_{\max,t} & \text{if} ~(u_l=u_k) \in \mathcal N_{\text{tra},t}~\text{or}~(u_l \notin \mathtt {u}(\mathtt {f}^{d2d}(\mathcal N_{\text{tra},t}))~\text{and}~u_k \in \mathcal N_{\text{tra},t})
\end{cases}
\end{align} 
\end{definition}

\begin{definition} The completion time of $u_l$-th user, denoted by $\mathtt T_{u_l}$, is the total time required in seconds to receive all its requested files. The overall completion time $\mathtt T_o$ is the time required to receive all files by all users, and is given by $\mathtt T_o=\max _{u_l \in\mathcal{N}_w}\{\mathtt T_{u_l}\}$.
\end{definition}

\begin{definition} A transmission schedule ${\mathcal{S}}=$ $\{( \mathtt f^{c}_{e_n,t},R^c_{e_n}), (u_k, \mathtt f^{d2d}_{u_k,t},r^{d2d}_{u_k})\}_{\forall t \in \{1,2,.......,\mathcal {j{\mathcal{S}}j}\}, \forall e_n\in\mathcal{K}, \forall u_k\in \mathcal N_{\text{tra},t}}$ is a collection of transmitting eRRHs/D2D transmitters, their file combinations and adopted rates at every $t$-th transmission index to receive all files by all users. 
\end{definition}

The completion time minimization
problem in a D2D-aided F-RAN system can be expressed
as follows
 \begin{align} \label{eq:Sopt1}
 \mathcal S^* &= \arg\min_{\mathcal S \in \mathbf{S}}\{\mathtt T_{o}(\mathcal S)\}    = \arg\min_{\mathcal S \in \mathbf{S}} \left \{\max_{u_l \in \mathcal N_w}\left\{ \mathtt T_{u_l}(\mathcal S)\right\}\right \},
 \end{align}
 where $\mathcal S^*$ is the schedule that optimally minimizes the overall completion time and $\mathbf S$ is the set of all possible
transmission schedules. The follwoing theorem expresses the  optimal schedule $\mathcal S^*$ in terms of time delay defined in definition 1.
\begin{theorem}\label{th:pmp}
The optimal schedule $\mathcal S^*$ that minimizes the overall
completion time in a D2D-aided F-RAN  system can be written as follows
 \begin{align} \label{eq:Sopt}
 \mathcal S^* = \arg\min_{\mathcal S \in \mathbf{S}} \left \{\max_{u_l \in \mathcal N_w}\left \{\frac{B.|\mathcal W_{u_l,0}|}{\tilde{R}_{u_l}(\mathcal S)} + \mathbb{D}_{u_l}(\mathcal S)\right \} \right \},
 \end{align}
where $|\mathcal W_{u_l,0}|$ is the initial \textit{Wants} size of $u_l$-th user, $\mathbb{D}_{u_l}(\mathcal S)$ is the accumulative time delay of $u_l$-th user in schedule, and $\tilde{R}_{u_l}(\mathcal S)$ is the harmonic mean of the  transmission rates of  transmissions  that   are instantly decodable for    $u_l$-th user in schedule $\mathcal S$.
\end{theorem}

\begin{proof}
The proof of  \thref{th:pmp} is omitted in this paper due to the space limitation\ignore{However, we can use  the same steps that proofed Theorem 1 in  \cite{21n} for C-RAN networks.  It is important to  note that Theorem 1 in  \cite{21n}  considered one transmission rate for all RRHs unlike  \thref{th:pmp} that different transmission rates and powers for  D2D-aided F-RAN system,  which makes the statistics of the harmonic mean $\tilde{R}_{u_l}(\mathcal S)$ different.} and only a sketch of the proof is given as follows. We first show that the completion time can be expressed as the sum of instantly and non-instantly decodable transmission times from $|\mathcal K|$ and $|\mathcal N_{\text{tra}}|$ transmitters via cellular and D2D links, respectively. Afterward, we need to proof that the number of instantly decodable transmissions to $u_l$-th user is equal to the number of its requested files $|\mathcal W_{u_l,0}|$ and
the number of non-instantly decodable transmissions matches
the time delay in definition 1. Finally, we extend the results of the optimal schedule in Theorem
1 in \cite{18n} that used in PMP system with a single transmitter to the coordinated D2D-aided F-RAN setting with multiple transmitters.
\end{proof}

Solving the completion time problem in \eref{eq:Sopt1} optimally is intractable \cite{21n}. In fact, the
transmission schedule at the current transmission slot does not depend only on the future transmission schedules, but also on users' achievable capacities and eRRHs' transmit powers. Therefore, we pay our special attention to solve such problem at each transmission, where  files are transmitted with high transmission rates. If some eRRHs cannot send XOR files to a set of users with the rate threshold $R_{\text{th}}$, these users can be scheduled on D2D links. To this end, our main objective is to minimize the completion time at each transmission, known as the anticipated  completion time \cite{16n}, through minimizing the time delay. This anticipated user's completion time at each transmission in D2D-aided F-RAN system is given in the next corollary.
\begin{corollary} \label{cor:LBcompletion}
The anticipated completion time of $u_l$-th user at $t$-th transmission slot is given by
 \begin{align}\label{9}
   {\mathtt T}_{u_l,t} \approx  \frac{B.|\mathcal W_{u_l,0}|}{\tilde{R}_{u_l,t}}  + \mathbb{D}_{u_l,t},
  \end{align}
where $\mathbb{D}_{u_l,t}$ is the accumulative transmission delay as given in \eref{eq3}, and  $\tilde{R}_{u_l,t}$ is the harmonic mean of the transmission rates that are instantly decodable for $u_l$-th user until $t$-th
transmission.
\end{corollary}
The anticipated completion time in \corref{cor:LBcompletion} depends on the number of requested files by $u_l$-th user, its accumulated time delay and harmonic mean $\tilde{R}_{u_l,t}$. Clearly, this metric is intimately related to the duration of time that all files are delivered to all users, which can be illustrated in the following example.\\
\begin{figure}[t!]
      \centering
      \begin{minipage}{0.494\textwidth}
          \centering
         \includegraphics[width=0.65\textwidth]{fig2nn.png} 
          \caption{Transmission time structure for  eRRHs and potential D2D transmitters for one time slot.}
          \label{fig2}
      \end{minipage}\hfill
      \begin{minipage}{0.494\textwidth}
          \centering
         \includegraphics[width=0.85\textwidth]{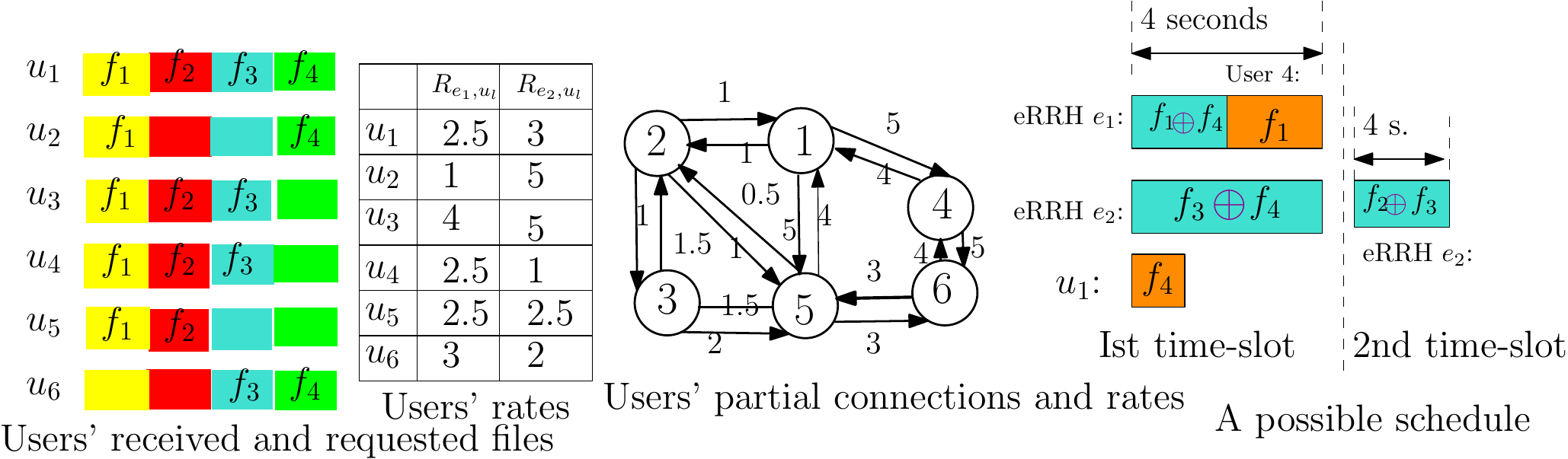} 
          \caption{D2D-aided F-RAN system containing $6$ users and their corresponding requested/received files and rates.}
          \label{fig3}
       \end{minipage}
  \end{figure} 
\ignore{  
\begin{figure}[t]
\centering
\includegraphics[width=0.5\linewidth]{./fig/fig1nnnn.pdf}
\caption{D2D-aided F-RAN system containing $6$ users and their corresponding requested/received files and rates. For example, $u_2$ receives $f_1$, $f_4$ and requests $f_2$, $f_3$. The sets of files that stored in eRRHs' caches are: $\mathcal C_{e_1}=\{f_1, f_4, f_3\}$, $\mathcal C_{e_2}=\{f_2, f_3, f_4\}$.}
\label{fig3}
\end{figure}}
\textbf{Example 1:} This example considers the model in  Fig. \ref{fig3} that consists of $2$ eRRHs, $6$ users, users' received and requested files and their rates.  For example, $u_2$ receives $f_1$, $f_4$ and requests $f_2$, $f_3$. The sets of files that stored in eRRHs' caches are $\mathcal C_{e_1}=\{f_1, f_4, f_3\}$, $\mathcal C_{e_2}=\{f_2, f_3, f_4\}$. Each file is assumed to have a size of $10$ bits. To minimize the completion time for this example, one possible schedule is given as follows.\\
\textbf{First time slot:} The $e_1$-th and $e_2$-th eRRHs transmit $\mathtt f^{c}_{e_1,1}=f_1 \oplus f_4$ and $\mathtt f^c_2=f_3 \oplus f_4$ with rates $R^c_{e_1}=2.5$ and $R^c_{e_2}=5$ bits/s, respectively, to the sets  $\mathtt u(\mathtt f^{c}_{e_1,1})=\{u_4,u_6\}$ and $\mathtt u(\mathtt f^{c,1}_{e_2})=\{u_2,u_3\}$. The $u_1$-th user transmits $\mathtt f^{d2d}_{u_1,1}=f_4$ with rate $r^{d2d}_{u_1}=5$ bits/s to the set $\mathtt u(\mathtt f^{d2d}_{u_1,1})=\{u_5\}$. Given this, we have the following transmission durations of $e_1$-th eRRH, $e_2$-th eRRH, and $u_1$-th transmitting user, respectively: $T^c_{e_1}=\frac{10}{2.5} = 4, T^c_{e_1}= \frac{10}{5} = 2, T^{d2d}_{u_1}=\frac{10}{5} = 2$ seconds.  Since user $u_4$ receives $f_4$ from $e_1$-th eRRH in $4$ seconds, it can use the idle time of $e_1$-th eRRH, i.e., $T^c_{e_1\text{idle}}=2$ seconds, to send $f_2$ to $u_6$ with rate $r^{d2d}_{u_4}=5$ bits/s. Therefore, the updated \textit{Wants} sets after the first time slot are: $\mathcal W_{u_2,1} = \{f_2\}$, $\mathcal W_{u_3,_1} = \varnothing, \mathcal W_{u_4,1} = \varnothing$, $\mathcal W_{u_5,1} = \{f_3\}$, $\mathcal W_{u_6,1} = \varnothing$. Note that $T_{\max,1}=\max(T^c_{e_1},T^c_{e_2},T^{d2d}_{u_1}) = 4$ seconds.\\
\textbf{Second time slot:} The $e_2$-th eRRH transmits $\mathtt f^{c,2}_{e_2}=f_2 \oplus f_3$ with rate $R^c_{e_2}=2.5$ bits/s to the set $\mathtt u(\mathtt f^{c,2}_{e_2})=\{u_2,u_5\}$ which requires transmission time  $T^c_{e_2}=T_{\max,2}=\frac{10}{2.5} = 4$ seconds. By the end of second time slot, all users will have their wanted files. Therefore, the total transmission time is $ T_{\max,1} + T_{\max,2}=8$ seconds.
\section{Problem Formulation and Problem Decomposition}\label{PFPD}
\ignore{In this section, we first formulate the completion
time minimization problem in D2D-aided F-RAN at every transmission $t$. Given the intractability of the problem, we decompose it
in Section \ref{PD}.}\ignore{For D2D transmissions, the power allocations of transmitting users are assumed to be pre-determined.}
\subsection{Problem Formulation}
In order to minimize the completion time at each transmission slot, we need to develop a rate-aware network
coding framework that decides: i) the adopted transmission
rate/power at the $e_n$-th eRRH, $\{R^c_{e_n}, P_{e_n}\}$, to transmit its XOR combination $\mathtt f^{c}_{e_n,t}$ to a set of targeted users $\mathtt u(\mathtt f^{c}_{e_n,t})$, $\forall e_n\in\mathcal K$, and ii) the set of D2D transmitters $\mathcal N_{\text {tra},t}$ for sending $\mathtt f^{d2d}_{u_k,t}$ to the users $\mathtt u(\mathtt f^{d2d,t}_{u_k})$, and their adopted transmission rates $r^{d2d}_{u_k}, \forall u_k\in \mathcal N_{\text{tra},t}$. As
such, all files are delivered to all users with minimum completion time. Therefore, the completion time minimization problem in D2D-aided F-RAN system can be formulated as
\begin{subequations}
\begin{align} \nonumber \label{eqn:Mschedule}
&\rm P1: \hspace{0.2cm} \min_{\substack{\mathtt f^{c}_{e_n,t}, \mathtt f^{d2d}_{u_k,t}, P_{e_n}, \mathcal N_{\text{tra},t}\in \mathcal P(\mathcal{N})
 }} \left \{\max_{u_l \in \mathcal N_{w,t}} \mathtt T_{u_l,t}\right \}\\ \nonumber
   &\rm subject~to 
 \begin{cases}
    \text{C1:} \hspace{0.1cm} \mathtt u(\mathtt f^{c}_{e_n,t}) \cap \mathtt u(\mathtt f^{c}_{e_{n^\prime},t}) =\varnothing, \forall e_n \neq {e_{n^\prime}} \in \mathcal K;\\
     \text{C2:}\hspace{0.1cm} \mathtt u(\mathtt f^{d2d}_{u_k,t}) \cap \mathtt u(\mathtt f^{d2d}_{u_{k^\prime},t}) =\varnothing \mbox{ \& } \mathtt u(\mathtt f^{d2d}_{u_k,t}) \cap \mathtt u(\mathtt f^{c}_{e_n,t})=\varnothing, \forall u_k \neq {u_{k^\prime}} \in \mathcal N_{\text{tra},t}, e_n\in \mathcal K;\\
      \text{C3:}\hspace{0.1cm} r^{d2d}_{u_k}.T^{c}_{{e_{n^\prime}\text{idle}}}\geq B, ~\forall u_k\in \mathcal N_{\text{tra},t},~\forall e_{n^\prime}\in \mathcal{K}; \\
       \text{C4:}\hspace{0.1cm}\mathtt f^{c}_{e_n,t}\subseteq \mathcal P( \mathcal{C}_{e_n}) \mbox{ \& } \mathtt f^{d2d}_{u_k,t}\subseteq \mathcal P( \mathcal{H}_{u_k,t}), ~\forall (e_n,u_k)\in \mathcal{K}\times\mathcal N_{\text{tra},t};  \\
       \text{C5:}\hspace{0.1cm} 0\leq P_{e_n}\leq P_{\max}, \forall e_n\in \mathcal{K};
        \text{C6:}\hspace{0.1cm} R^c_{e_n}\geq R_{\text{th}}; \text{C7:}\hspace{0.1cm} r^{d2d}_{u_k}\geq R_{\text{th}}, \forall e_n\in \mathcal{K}, \forall u_k\in \mathcal N_{\text{tra},t}. 
\end{cases}
       \end{align}
\end{subequations}
\ignore{where $\mathcal P(\mathcal N)$ denotes the power set of the D2D transmitters.} The constraints are explained as follows. C1 states that the set of scheduled users to all eRRHs
are disjoint, i.e., each user must be scheduled to only one eRRH.  C2  makes sure that each user can be scheduled to only one potential D2D transmitter and no user can be scheduled to a D2D transmitter and eRRH at the same time instant. C3 ensures the successful delivery of files from D2D transmissions within the idle time of the eRRHs. C4 ensures that all files to be combined using
XOR operation at all eRRHs and D2D transmitters are stored in their \textit{Caches} and \textit{Has} sets, respectively.\ignore{ C5 says that each user in the network has the opportunity to send coded file over D2D links except users in the set $\mathtt u^c_e$.}  C5 bounds the maximum transmit power of each eRRH, and C6 and C7 satisfy the minimum transmission rates required to meet the QoS rate requirement $R_{\text{th}}$.\ignore{ If some eRRHs cannot meet this rate threshold $R_{\text{th}}$, they will be offloaded and users will be scheduled on D2D links.}

The optimization problem in P1 contains the NC scheduling parameters $\mathtt u(\mathtt f^{c}_{e_n,t}), \mathtt u(\mathtt f^{d2d}_{u_k,t}), \forall e_n\in \mathcal K, \forall u_k\in \mathcal N_{\text{tra},t}$, power allocations of eRRHs  $P_{e_n},~\forall e_n\in \mathcal K$, potential set of transmitting users $\mathcal N_{\text{tra},t}$ and their transmission rates.  We can readily show that problem P1 is NP-hard and intractable \cite{40}.   However, by analyzing the problem, we can decompose it into two subproblems and solve them individually  and efficiently using graph theory technique \cite{32}.

\subsection{Problem Decomposition}\label{PD}
Since the main objective is to minimize the maximum completion time of users, which depends on the time delay increase at each transmission slot, we can first focus on minimizing the transmission duration for the eRRH-user NC transmissions. In particular, we can get the possible completion time by jointly optimizing the NC user scheduling and power allocations of eRRHs. The mathematical formulation for minimizing the transmission duration for eRRH-user NC transmissions can be expressed as
\begin{subequations}
\begin{align} \label{eqn:Mschedule}
&\rm P2: \min_{\substack{ 0\leq P_{e_n}\leq P_{\max}}} T^c_{e_n,t} \\ \nonumber
   &{\rm subject ~to\ } 
   \begin{cases}
   \mathtt u(\mathtt f^{c}_{e_n,t})\cap \mathtt u(\mathtt f^{c}_{e_{n^\prime},t}) =\varnothing, \forall e_n \neq {e_{n^\prime}} \in \mathcal K;\nonumber\\
 \mathtt f^{c}_{e_n,t}\subseteq \mathcal P( \mathcal{C}_{e_n}); \hspace{0.1cm} R^c_{e_n}\geq R_{\text{th}}, ~\forall e_n\in \mathcal{K}. \nonumber
       \end{cases}
      \end{align}
\end{subequations}
Note that this subproblem contains users' associations and power allocation variables and a joint solution will be developed in \sref{JS}.

After obtaining the possible transmission duration from eRRH-user NC transmissions, denoted by $T^{c}_{e_{n^\ast},t}$ of $e_{n^\ast}$-th eRRH, by solving P2, we can now formulate the second subproblem. In particular, we can maximize the number of users $Z_t$ that are not been scheduled to the eRRHs $\mathcal N_{w,t}\backslash \mathtt u(\mathtt f^{c}_{e_{n^\ast},t})$ within $T^{c}_{e_{n^\ast},t}$ by using D2D communication. In addition, users being scheduled to  $e_{n^\prime}$-th eRRH from subproblem P2, have the opportunity to be scheduled on D2D links within the idle times of their corresponding eRRHs at the $t$-th transmission slot, $\forall e_{n^\prime}\neq e_{n^\ast}\in \mathcal K$. Therefore, the second subproblem of maximizing the number of users to be scheduled on D2D links can be expressed as follows
\begin{subequations}
\begin{align} \label{eqn:Mschedule}
&\rm P3: \max_{\substack{ \mathcal N_{\text{tra},t}\in \mathcal P(\mathcal{N}\backslash \mathtt u(\mathtt f^{c}_{e_{n^\ast},t}))\\ \mathtt f^{d2d}_{u_k,t}\subseteq \mathcal P(\mathcal{H}_{u_k,t})}} Z_t \\ \nonumber
   &{\rm subject ~to\ } 
\begin{cases}
     (\text{C2});     r^{d2d}_{u_k}.(T^{c}_{e_{n^\ast},t}-T^c_{e_{n^\prime},t})\geq B, ~\forall u_k\in \mathcal N_{\text{tra},t},~\forall (e_{n^\ast}, e_{n^\prime})\in \mathcal{K}; \\
     T^{d2d}_{u_k}\leq T^{c}_{e_{n^\ast},t}, ~\forall u_k\in \mathcal N_{\text{tra},t}; \\
     |\mathtt {u}(\mathtt {f}^{d2d}(\mathcal N_{\text{tra},t}))|+|\mathcal N_{\text{tra},t}| \leq Z_t.
\end{cases}
\end{align}
\end{subequations}

The constraint C3 in P1 is rewritten as the second constraint in P3 since we know $T^{c}_{e_{n^\ast},t}$. Fourth constraint  states that the transmission duration of any D2D transmitter should be less than or equal to $T^{c}_{e_{n^\ast},t}$. The last constraint is the maximum number limitation of scheduled users on D2D links. It can be easily observed that P3 is a
D2D scheduling problem that considers selection of D2D transmitters, their NC files and transmission rates.\ignore{ In the next section, we propose a graph-based approach to transform and solve the decomposed subproblems in P2 and P3 efficiently. }

\ignore{\begin{table*}
\vspace*{-0.8cm}
\begin{normalsize}  
\begin{subequations}
\begin{align} \label{eqn:Mschedule}
&\rm P3: \max_{\substack{ \mathcal N_\text{tra}^t\in \mathcal P(\mathcal{N}\backslash \mathtt u(\mathtt f^{c,t}_{e_n}))\\ \mathtt f^{d2d,t}_{u_k}\subseteq \mathcal P(\mathcal{H}^t_{u_k})}} Z^t \\ \nonumber
   &{\rm subject ~to\ } 
\begin{cases}
     (\text{C2});     r^{d2d}_{u_k}.(T^{*,t}_{e_n}-T^c_{e_{n^\prime}\text{wait}})\geq B, ~\forall u_k\in \mathcal N_\text{tra}^t,~\forall (e_n, e_{n^\prime})\in \mathcal{K}; \\
     T^{d2d}_{u_k}\leq T^{*,t}_{e_n}, ~\forall u_k\in \mathcal N_\text{tra}^t; \\
     \sum _{u_k\in \mathcal N_\text{tra}^t} |\mathtt u(\mathtt f^{d2d,t}_{u_k})|+|\mathcal N_\text{tra}^t| \leq Z^t.
\end{cases}
\end{align}
\end{subequations}
\end{normalsize}
\vspace*{-0.8cm}
\hrulefill
\end{table*}}
\section{Completion Time Minimization:  Joint Approach}\label{JA}
In this section, we propose a  joint approach to 
solve the subproblems in P2 and P3 using designed interference-aware IDNC and new D2D conflict graphs, respectively. Specifically, we design interference-aware IDNC in the first subsection to solve the subproblem P2\ignore{ by optimizing the network-coded scheduling and power level of each eRRH, as shown} in the second subsection. We then introduce a new D2D conflict graph to solve the subproblem P3 as shown in the third and fourth subsections, respectively.

\ignore{This section proposes an efficient joint approach that consists of two sequential solutions for solving the subproblems in P2 and P3 using designed interference-aware IDNC and new D2D conflict graphs, respectively.}

\subsection{Subproblem P2 Transformation: Interference Aware-IDNC Graph}\label{IA}
Interference-Aware IDNC (IA-IDNC) graph, denoted by $\mathcal{G}_\text{IA-IDNC}(\mathcal V, \mathcal E)$, is designed to systematically select an IDNC combination, transmission rate, and power allocation of each eRRH at the $t$-th transmission slot. Unlike the graph in \cite{21n} that resulted in one rate for fixed power eRRHs, our designed IA-IDNC graph leads to different transmission rates/powers from different eRRHs. This gives flexibility to each eRRH to choose its IDNC combination and transmission rate that satisfy a set of scheduled users. 

Consider  generating all possible associations (pairs) representing users and their corresponding requested files that cached by $e_n$-th eRRH, denoted by $\mathcal A_{e_n}=\mathcal N_{w}\times \mathcal C_{e_n}$, i.e., $a\in \mathcal A_{e_n}=(u_l,f_h)$ represents the association of $u_l$-th user and its $f_h$-th requested file. The corresponding files of a set of associations in $\mathcal A_{e_n}$ can be encoded into one IDNC combination if these files are instantly decodable to the corresponding associated users. The set of all IDNC combinations is denoted by $\mathcal A_{e_n,\text{IDNC}}$. In particular, the corresponding files of any two different associations $a\in \mathcal A_{e_n}$ and $a^\prime \in \mathcal A_{e_n}$ are encoded if one of the following IDNC conditions is satisfied.
\begin{itemize}
\item \textbf{IDNC-C1:} $u_{l,a} \neq u_{l^\prime,a^\prime}$ and $f_{h,a} = f_{h,a^\prime}$ This condition represents that the same file $f_h$ is requested by two distinct  users $u_l$ and ${u_{l^\prime}}$. 
\item \textbf{IDNC-C2:} $u_{l,a} \neq u_{l^\prime,a^\prime}$ and $ f_{h^\prime,a^\prime} \in  \mathcal H_{u_l,a}$ and $f_{h,a} \in  \mathcal H_{u_{l^\prime},a^\prime}$. This condition represents that different files $f_{h^\prime}$ and $f_{h}$ are requested by two different users $u_{l^\prime}$ and $u_{l}$, respectively. Meanwhile, the requested file of each user is in the \textit{Has}
set of the user in the other association. We use $l,a$ in $u_{l,a}$ as subscripts to represent $u_l$-th user in
$a$-th association.
\end{itemize}
For example, the element $\mathtt a=(\mathtt f^{c}_{e_n},\mathtt u(\mathtt f^{c}_{e_n}))\in \mathcal A_{e_n,\text{IDNC}}$ represents the set of scheduled users $\mathtt u(\mathtt f^{c}_{e_n})$  that will receive the IDNC combination $\mathtt f^{c}_{e_n}$ from $e_n$-th eRRH.

Let $\mathcal{S}_{e_n}$ be the set of all possible associations between the IDNC combinations $ \mathcal A_{e_n,\text{IDNC}}$ and achievable capacities $\mathcal{R}_{e_n}\subset \mathcal{R}$, i.e., $\mathcal{S}_{e_n}=  \mathcal A_{e_n,\text{IDNC}}\times\mathcal{R}_{e_n}$. In other words, $\mathbf S=(\mathtt a, R)\in \mathcal{S}_{e_n}$ is a schedule that consists of a set of associations representing the IDNC combination, set of scheduled users, and rate  $R$ of $e_n$-th eRRH, i.e., $\mathbf S= (\mathtt a, R)= s_1,s_2,\cdots, s_{|\mathbf S|}$, where $|\mathbf S|$ is the total number of scheduled users in $\mathbf S$. Note that $s_1$ represents one user, one file, and rate of $e_n$-th eRRH\ignore{, and accordingly, each association $s^j_{e_n,i}\in \mathbf S_{e_n,i}$ is represented by $s_{e_n,R_{e_n},u_l,f_h}$. All feasible schedules $\mathbf S_{e_n,i}$ of the $e$-th eRRH  is represented by the set $\mathcal S_{e_n,t}$}. 
Now, any two
associations $s_1 \in \mathbf S, s_2\in  \mathbf{S}$  representing the $e_n$-th eRRH should have an equal adopted rate that is greater than or equal to $R_\text{th}$. That is, the \textbf{Rate Condition (RC)} is satisfied $ R_{s_1} = R_{s_2}$ and $R_{s_1}\geq R_\text{th}$. 

The aforementioned procedures are applied to all eRRHs in the network. Thus, the set of all possible IDNC combinations $\mathcal{A}_{e_n,\text{IDNC}}$ and schedules $\mathcal S_{e_n}$ in the network are $\mathcal{A}_{\text{IDNC}} = \bigcup\limits_{e_n\in \mathcal{K}}\mathcal{A}_{e_n,\text{IDNC}}$ and $\mathcal{S} = \bigcup\limits_{e_n \in \mathcal{K}}\mathcal{S}_{e_n}$, respectively. These schedules $\mathcal{S}$ can be exactly represented by unique vertices $ \mathcal V$ in $\mathcal G_\text{IA-IDNC}(\mathcal V,\mathcal E)$ such that
we transfer the subproblem P2 to a graph-theory based problem.\ignore{ Graph theory techniques are sufficient with
numerous proven theorems that capable of efficiently analyzing large graphs representing complex
problems \cite{GT}.} Therefore, the $\mathbf S_{i}$-th schedule in $\mathcal{S}$ is represented by the $V_i$-th unique vertex  in $\mathcal G_{\text{IA-IDNC}}$ ($i=1,2,\cdots,|\mathcal S|$). This schedule-to-vertex mapping makes any IDNC combination sent from  the $e_n$-th eRRH with adopted rate to its corresponding associated users  is decodable.

Two vertices $V_i$ and $V_{i\prime}$ representing two different schedules $\mathbf{S}_{i}\in \mathcal S_{e_n}$ and $\mathbf{S}_{i^\prime} \in \mathcal{S}_{e_{n^\prime}}$ are adjacent by an edge in $\mathcal G_{\text{IA-IDNC}}$,  if the associations they represent satisfy the
following condition.
\begin{itemize}
 \item  \textbf{Transmission Condition (TC):} $\mathtt u\cap \mathtt u'=\phi$, $\forall (s_1,s_2)\in \mathbf S_i \times \mathbf S_{i^\prime}$. This condition ensures that the same user can be scheduled only to a unique eRRH.
\end{itemize}
Assuming that the power allocation of the eRRHs in the
network will be computed later; then the weight of a given vertex $V$ representing   a schedule $\mathbf{S}$ is expressed by
\begin{equation}
\label{eq10}
w(V)=\sum_{s\in \mathbf{S}}\frac{\min_{u_{l,s}\in \mathtt u(\mathtt f^{c}_{e_n})}\log_{2}(1+\text{SINR}_{e_{n,s},u_{l,s}}(\textbf{P}))}{B}\ignore{~ \forall k=1,2,\cdots, |\mathtt u(\mathtt f^{c}_{e_n,t})|,}.
\end{equation}
The weight of each vertex reflects
the contribution of each eRRH towards minimizing the completion time of its associated users. Actually, the transmission rate plays a crucial role in minimizing the transmission duration $T^c_{e_n}$.\ignore{for transmitting  the IDNC file $\mathtt f^{c}_{e_n}$ to set of users $\mathtt u(\mathtt f^{c}_{e_n})$. This is equivalent to schedule users $\mathtt u(\mathtt f^{c}_{e_n})$ to $e_n$-th eRRH with the maximum possible
transmission rate in  \eref{eq10}.} Thus, a larger value in \eref{eq10}  leads to minimize the transmission durations of delivering IDNC files to users, which minimizes their completion times. 

The design of the IA-IDNC graph makes any  maximal weight clique\ignore{\footnote{A maximal clique in $\mathcal G_\text{IA-IDNC}$ is a clique that consists of pairwise adjacent vertices and cannot be extended to
include one more vertex without violating the pairwise adjacent conditions. The 
maximum weight clique is the maximal clique with the largest sum of its vertices’ weights.}} represents a set of transmissions satisfying the following three features: i) each user is scheduled only to a single eRRH that cached one of its requested files,
ii) each eRRH delivers an IDNC file with an adopted transmission rate/power that satisfies a lower completion time for a potential set of users. Such  adopted rate satisfies the QoS rate guarantee and no larger than the channel  capacities of all scheduled users,
iii)  the  weight of each vertex strikes a balance between  the adopted rate and the number of scheduled  users to each eRRH.

The following theorem characterizes
the solution of subproblem P2 based on the designed IA-IDNC graph.
\begin{theorem} \label{thm:Mweight}
 The transmission duration minimization  subproblem P2 is equivalently represented by the maximum weight
 clique problem
 in the IA-IDNC graph, and can be expressed as 
 \begin{align}
 \ignore{& \arg \max_{\substack{ \\ \mathtt f^{c}_{e_n,t} \in \mathcal{P}(\mathcal{C}_{e_n}) \\ R^c_{e_n} \in \mathcal{R}_t \\ P_{e_n} \in \{0,...,P_{\max}\}}} \sum_{e_n \in \mathcal K}  \sum_{s_{e_n,R^c_{e_n},u_l,f_h}\in  \textbf{S}_{e_n,i}}\frac{\min_{u_l\in \mathtt u(\mathtt f^{c,t}_{e_n}(\textbf{S}_{e_n,i}))}\log_{2}(1+\text{SINR}_{e_n,u_l}(\textbf{P}))}{B}\\
 &=   \arg\max_{\mathtt S_t \in \mathcal S_t} \sum_{e_n \in \mathcal K} \sum_{\textbf{S}_{e_n,i} \in \mathtt S_t} w (\textbf{S}_{e_n,i})} =\arg\max_{\mathtt C \in \mathcal C} \sum_{V_i \in \mathtt C} w (V_i),~ \forall i=1,2,\cdots, |\mathtt C|,
 \end{align}
 where $\mathtt C$ is the maximum weight clique of a maximum degree $|\mathcal K|$ in the IA-IDNC graph and $\mathcal C$ is the
 set of all possible maximal cliques.\ignore{ In \cite{21n}, the optimization
 becomes
  $  \arg \max_{\substack{ \\ \mathtt f^{c,t}_{e_n} \in \mathcal{P}(\mathcal{F}) R_{e_n} \in \mathcal{R}^t}} \sum_{e_n \in \mathcal K} \sum_{u_l \in \mathtt u(\mathtt f^{c,t}_{e_n})} \frac{R_{\min}}{B},  $ where $R_{\min}$ is the minimum transmission rate among all eRRHs.}
 \end{theorem}

 \ignore{\begin{lemma} \label{lem:Fdecisions}
 All feasible schedules in $\mathcal S^t$ of transmitting eRRHs, their transmitted IDNC files to sets of scheduled users, and their transmission rates/powers in the $t$-th time slot are defined by the set of all maximal cliques in the $\mathcal G^t_\text{IA-IDNC}$ graph.
 \end{lemma}}
 
\begin{proof}
The proof of \thref{thm:Mweight} is omitted due to space limitation and only we provide a sketch of it as follows. First, we need to show that there is a unique  one-to-one mapping between each schedule $\mathbf S_{i}\in \mathcal S$ and each vertex $V_i\in \mathcal V$ in $\mathcal G_\text{IA-IDNC}$ ($i=1,2,\cdots, |\mathcal S|$). Then, we  emphasize that
each maximal clique in the IA-IDNC graph that consists of a set of vertices  satisfying all edge conditions represents feasible coded transmissions from the eRRHs. Finally, the proof can be concluded by
 showing that the contributed weight of the maximum weight clique $\mathtt C$ for minimizing the transmission duration is: 
  $ w (\mathtt C)=\sum_{V_i \in \mathtt C} w (V_i)= \sum_{e_n \in \mathcal K} \sum_{\mathbf{S}_{i} \in \mathtt S} w (\mathbf{S}_{i})=\sum_{e_n \in \mathcal K}  \sum_{s\in  \mathbf{S}_{i}}\frac{\min_{u_{l,s}\in \mathtt u(\mathtt f^{c}_{e_n})}\log_{2}(1+\text{SINR}_{e_{n,s},u_{l,s}}(\textbf{P}))}{B},~ \forall i=1,2,\cdots, |\mathtt C|$,
 where $\mathtt S$ is the set of the selected potential schedules. Therefore, the subproblem P2 is equivalent to the maximum weight  clique problem among the maximal cliques in the IA-IDNC
 graph. 
 \end{proof}
  
The problem in \thref{thm:Mweight} is clearly NP-hard problem, and solving it optimally is intractable \cite{NP}. However, we heuristically  and efficiently solve it in the next subsection.

\subsection{Solution of Subproblem P2}\label{JS}
In this section, we solve the problem in \thref{thm:Mweight} by characterizing  the joint solution to the NC user scheduling and power allocation problem while designing the IA-IDNC graph. A proper power allocation for each eRRH leads to  suppress the interference in the system, thus a better transmission rate is achieved. As a result, the transmission duration for delivering files to the scheduled users is minimized.

Consider a power-clique (PC) in $\mathcal G_\text{IA-IDNC}$ that is associated with a network-coded user scheduling $\mathtt S=\{ \mathbf{S}_{1}, \mathbf{S}_{2},\cdots, \mathbf{S}_{|K|}\}$, where $\mathbf{S}_{1}$ is the schedule of $e_1$-th eRRH, which consists of a set of associations $s_1,s_2,\cdots, s_{|\mathbf S_1|}$, and $|\mathbf S_{1}|=|\mathtt u(\mathtt f^{c}_{e_n})|$. Our goal
is to obtain a local optimal eRRH power allocation 
vector, denoted as  ($P^{*}_{e_1},\cdots,P^{*}_{e_K}$) for that PC. The power allocation problem
is formulated as an optimization problem of maximizing the weighed sum-rate. As such, all scheduled users receive
files sent by their associated eRRHs with minimum transmission duration, which can be expressed as 
\begin{align}
\label{eq13} \nonumber
&\text{P4}: \max_{P_{e_n}}\sum _{n=1}^{K} \frac{|\mathtt u(\mathtt f^{c}_{e_n}(\mathbf{S}_n))|}{B} * \min_{u_l \in \mathtt u(\mathtt f^{c}_{e_n}(\textbf{S}_{n}))} \left(\log_{2}(1+\text{SINR}_{e_n,u_l}(\textbf{P}))\right),\\ 
&~{\rm s.~t.\ } 0\leqslant P_{e_n}\leqslant P_{\max},\forall n=1,2,\cdots, K,
\end{align}
where $\mathtt u(\mathtt f^{c}_{e_n}(\mathbf{S}_{n}))$ is the set of scheduled users in $n$-th schedule corresponding to $e_n$-th eRRH. The power allocation problem in P4 is a non-convex optimization problem. Therefore, similar to the works in literature (see for example, \cite{33, 34} and references therein), we focus on finding the local optimal solution.

The proposed solution to the problem in \thref{thm:Mweight} is executed at the CBS at each transmission slot and divided into two stages: designing the IA-IDNC graph and finding its corresponding maximum PC.

\textit{First stage:} The  IA-IDNC  graph is designed as follows.  Using \textbf{IDNC-C1}, \textbf{IDNC-C2}, and \textbf{RC} conditions that explained in \sref{IA},
we generate all schedules $\mathcal S$ and represent them by vertices in $\mathcal V$. Afterwards, we check the connection between any two pairs $V_i$ and $V_{i'}$ of vertices in  $\mathcal G_\text{IA-IDNC}$ based on the transmission condition \textbf{TC} in \sref{IA}. Any connected vertices result in
a feasible network-coded scheduling to the eRRHs. Then, we evaluate the power allocation of
such network-coded scheduling $\{ \mathbf{S}_{1}, \mathbf{S}_{2},\cdots, \mathbf{S}_{K}\}$ by solving the optimization problem \eref{eq13}. By the computed power allocation and corresponding rate, we compute the weights of the corresponding vertices in $\mathcal G_\text{IA-IDNC}$
as expressed in \eref{eq10}. We repeat the above  process to all network-coded schedules in the IA-IDNC graph. 

\textit{Second stage:} The second stage finds the maximum weight PC
$\mathtt C$ among all other maximal PCs  in $\mathcal G_\text{IA-IDNC}$
graph. In the first step, we select vertex $V_i\in \mathcal V, (i=1,2,\cdots, | \mathcal V|)$  that has the maximum weight $w (V^*_i)$  and add it to $\mathtt C$ (at this point $\mathtt C = \{V^*_i\}$). Then,  the  subgraph $\mathcal G_{\text{IA-IDNC}}(\mathtt C)$, which consists of vertices in graph  $\mathcal G_\text{IA-IDNC}$ that are adjacent to vertex $V^*_i$ is extracted and considered for the next vertex selection process. In the next step, a  new maximum weight  vertex $V^*_{i'}$ (i.e., $V^*_{i'}$ should be in the corresponding PC of $V^*_i$) is selected from subgraph $\mathcal G_{\text{IA-IDNC}}(\mathtt C)$. Now, $\mathtt C = \{V^*_i, V^*_{i'}\}$. We repeat this process until no further vertex is adjacent  to all vertices in the maximal weight PC  $\mathtt C$. The selected $\mathtt C$ contains at most $K$ vertices.

\ignore{\begin{figure}[t]
\centering
\includegraphics[width=0.6\linewidth]{./figs/fig4_nnnN.pdf}
\caption{The IA-IDNC graph corresponding to  the system in Example $1$ and Fig. \ref{fig3} based on \sref{IA}. Every vertex represents a possible NC combination that consists of combined feasible associations in each eRRH. For plotting simplification, we did not include the uncoded (without NC) vertices that represents one association per vertex in the above graph.}
\label{fig4n}
\end{figure} }
\subsection {New D2D Conflict Graph }\label{D2D}
We introduce a new D2D conflict graph, denoted by $\mathcal{G}_\text{d2d}(\mathcal V, \mathcal E)$, that considers all possible conflicts for scheduling users on D2D links, such as transmission, network coding, half-duplex conflicts. This leads to feasible transmissions from  the potential D2D transmitters $|\mathcal N_{\text{tra}}|$, where each $u_k\in \mathcal N_{\text{tra}}$ transmits the IDNC combination $\mathtt f^{d2d}_{u_k}$ to the scheduled users  $\mathtt u(\mathtt f^{d2d}_{u_k})$ with the transmission rate $r^{d2d}_{u_k}$. 

Let $\mathcal N_{\text{d2d}}$ denote the set of unscheduled users to the eRRHs, i.e., $\mathcal N_{\text{d2d}}=\mathcal N\backslash\bigcup^{K}_{{n=1}} \mathtt u(\mathtt f^{c}_{e_n})$, and let $\mathcal N_{\text{d2d},w}\subset \mathcal N_{\text{d2d}}$ denote the set of users that still wants some files. Hence, the D2D conflict graph is designed by generating all vertices for $u_k$-th possible D2D transmitter, $\forall u_k \in \mathcal N_{\text{d2d}}$. The vertex set $\mathcal V$ of the entire graph is the union of vertices of all users.
Consider, for now, generating the vertices  of $u_k$-th user. Note that  $u_k$-th D2D transmitter can encode its IDNC file $\mathtt f^{d2d}_{u_k}$ using  its previously received files $\mathcal H_{u_k}$. Therefore, each
vertex is generated for each single file $f_h\in \mathcal W_{u_i}\cap \mathcal H_{u_k}$ that is requested by each user $u_i\in \mathcal N_{\text{d2d},w}\cap \mathcal Z_{u_k}$ and for each achievable rate $r$ of $u_k$-th user that is defined below.
\begin{definition}
The set of achievable rates $\mathcal R^{d2d}_{u_k,u_i}$ from $u_k$-th user to $u_i$-th user is  a subset of achievable rates $\mathcal R^{d2d}_{u_k}$ that are less than or equal to channel  capacity  $r^{d2d}_{u_k,u_i}$. It can be expressed by  $\mathcal R^{d2d}_{u_k,u_i} = \{r \in \mathcal R^{d2d}_{u_k}| r \leq r^{d2d}_{u_k,u_i}~ \text{and} ~u_i\in \mathcal N_\text{d2d,w}\cap\mathcal Z_{u_k}~ \text{and}~ r\geq R^c_{e_{n^\ast}} \}$.
\end{definition}
The above definition emphasizes that $u_i$-th user in the coverage zone $ \mathcal Z_{u_k}$ can  receive a file from $u_k$-th D2D transmitter if the adopted  transmission rate $r$ is in the achievable set $R^{d2d}_{u_k,u_i}$ and  no less than the minimum transmission rate of $e_{n^\ast}$-th eRRH. Therefore, we generate $|\mathcal R^{d2d}_{u_k,u_i}|$ vertices for a requesting  file  $f_h \in \mathcal H_{u_k}  \cap \mathcal W_{u_i}, \forall u_i \in  \mathcal N_\text{d2d,w}\cap \mathcal Z_{u_k}$. In summery,  a vertex $v_{r,i,f}^k$ is generated  for each association of a transmitting user $u_k$, a  rate  $r \in \mathcal R^{d2d}_{u_k,u_i}$, and  a requesting file  $f_h \in \mathcal H_{u_k}  \cap \mathcal W_{u_i}$ of user $u_i \in  \mathcal N_\text{d2d,w}\cap \mathcal Z_{u_k}$. Similarly, we generate all vertices  for all users in $\mathcal N_{\text{d2d}}$.\footnote{For the space limitation, we generate only the vertcies of $\mathcal N_{\text{d2d}}$ users and ignore those representing scheduled users to $e_{n^\prime}$-th eRRH, $ e_{n^\prime} \in \mathcal K$.  However, they can be generated and connected to each other using similar steps in this section with the difference that any D2D transmitter should be able to deliver files within the idle time slot of the corresponding eRRH.}

All possible conflict connections  between vertices (conflict edges between circles) in the D2D conflict graph are provided as follows. Two vertices $v_{r,i,h}^k$ and $v_{r',i',h'}^{k'}$ are adjacent   by a conflict edge in $\mathcal G_\text{d2d}$,  if one of the  following conflict conditions (CC) is true:
\begin{itemize}
\item  \textbf{IDNC (CC1):} ($u_k=u_{k^\prime}$) and ($f_h\neq f_{h^\prime}$) and ($f_{h},f_{h^\prime}$) $\notin \mathcal H_{u_{k^\prime}}\times \mathcal H_{u_{k}}$. A conflict edge between vertices is connected as long as the files they represent are not-instantly decodable to a set of scheduled users to the same $u_k$-th D2D transmitter.
 \item  \textbf{Rate (CC2):}  ($u_k = u_{k^\prime}$) and ($r\neq r'$). All adjacent vertices  correspond to  the same $u_k$-th D2D transmitter should have the same achievable rate.
 \item  \textbf{Transmission (CC3):}  ($u_k \neq u_{k^\prime}$) and ($u_i=u_{i^\prime}$). The same user cannot be scheduled to two different D2D transmitters $u_k$ and $u_{k^\prime}$. 
 \item \textbf{Half-Duplex (CC4):} ($u_k = u_{i^\prime}$) or ($u_{k^\prime}=u_i$). The same user cannot transmit and receive in the same  transmission slot.
 \end{itemize}
 
\ignore{Next, connecting two vertcies $v_{r,k,l}^i$  and $v_{R,k',l'}^b$  if and only if
\begin{itemize}
 \item  \textbf{G1}:  $k = k'$ and $l \neq l'$ we have $r \geq R_{b'\text{wait}}$. This guarantees that the same device $k$ can be targeted from device $i$ with a new file as long as its transmission rate can deliver the packet in the remaining time slot of that eRRH.
\ignore{\item  \textbf{G2}: $k \neq k'$. This condition completes the adjacencies in the graph for any two vertices not opposing the CC1 constraint for any two different users.}
\end{itemize}}
Given the aforementioned designed D2D conflict graph, the following theorem reformulates the subproblem P3.

\begin{theorem} \label{thm:Mweight1}
The subproblem of maximizing the number of scheduled users on D2D links in $P_3$ at the $t$-th transmission  is equivalently represented by the maximum weight
independent set (IS) selection  among all the maximal sets
in the $\mathcal G_\text{d2d}$ graph, where the weight $\psi(v_{r,i,h}^k)$ of each vertex $v_{r,i,h}^k$ is given by
\begin{align} \label{eq12}
 \psi(v_{r,i,f}^k) =|\mathcal Z_{u_k}\cap \mathcal N_{d2d,w}(\mathcal H_{u_k})|(\frac{r}{B}).
\end{align}
\end{theorem}
The above weight metric shows two potential benefits: i) $|\mathcal Z_{u_k}\cap \mathcal N_{d2d,w}(\mathcal H_{u_k})|$ represents   that the $u_k$-th transmitting user is connected to many other users that are requesting files in $\mathcal H_{u_k}$ and ii) $\left(\frac{r}{B}\right)$  provides a balance between  the transmission rate and the number of scheduled  users on D2D links.

\subsection{Details of Greedy Solution}\label{DGS}
The proposed solution here greedily maximizes the number of scheduled users on D2D links within $T^{c}_{e_{n^\ast}}$ by maximizing the number of vertices in any IS in the D2D conflict graph. In order to
maximize the number of vertices in any IS, we update the weight of each vertex. An appropriate design of the updated weights of vertices leads to selection of a large number of vertices and each vertex has high original weight that defined in \eref{eq12}. 

Let $\pi_{v,v'}$\footnote{For notation simplicity, we replace $v_{r,k,l}^i$ by $v$ and $v_{r',k',l'}^{i'}$ by $v'$.
} define the non-adjacency indicator of vertices $v$ and $v'$ in  the $\mathcal{G}_\text{d2d}$ graph such that:
 \begin{equation}
 \pi_{v, v'} =
   \begin{cases}
    1 & \text{if $v$ is not adjacent to $v'$ in $\mathcal{G}_\text{d2d}$}, \\
    0 & \text{otherwise}.
   \end{cases}
 \end{equation}
Now, let the weighted degree $n_{v}$ of vertex $v$   is defined by $n_{v} = \sum_{v' \in \mathcal{G}_\text{d2d}} \pi_{v, v'}. \psi(v')$, where $\psi(v')$ is the original weight  of vertex $v'$  defined in \pref{thm:Mweight1}.   Hence, the modified weight  of vertex $v$ is defined as
\begin{align}\label{eq15}
w (v) & = \psi(v) n_v  = \psi(v)\sum_{v'  \in \mathcal{G}_\text{d2d}} \pi_{v, v'}. \psi(v').
\end{align}

The modified weight of a vertex $v$ in \eref{eq15} points two
attractive features: (i) it has a large original
weight, and (ii) it is not adjacent to a large
number of vertices that have high original weights. Based on this, we iteratively and heuristically execute a greedy vertex search scheme as follows. Initially, a vertex $v^*$ that has the maximum weight $w(v^*)$ is selected  and added to the maximal IS $\Gamma$ (i.e., $\Gamma = \{v^*\}$). Then,  the  subgraph $\mathcal G_\text{d2d}(\Gamma)$, which consists of vertices in graph  $\mathcal G_\text{d2d}$ that are not adjacent to vertex $v^*$, is extracted and considered for the next  process. In the next step, a  new maximum weight  vertex $v'^*$ is selected from subgraph $\mathcal G_\text{d2d}(\Gamma)$ (at this point $\Gamma = \{v^*, v'^*\}$). We repeat this process for all transmitting users so that no further vertex is not adjacent  to all the vertices in $\Gamma$. The  selected D2D transmitters in the maximum IS  $\Gamma$  generate coded  files and broadcast  them to all neighboring users on D2D links.

The overall two-solution joint approach that are explained in \sref{JS} and \sref{DGS}, respectively, is provided in \algref{alg:LGS}.

\begin{algorithm}[t!]
\begin{algorithmic}[1]
\STATE Require $\mathcal N$, $\mathcal F$, $\mathcal K$,  $\mathcal C_{e_n}$, $\mathcal H_{u_l,0}$, $\mathcal W_{u_l,0}$, $P_{\max}$, $Q_{u_k}$, $B$, $h^c_{e_n,u_l}, h^{d2d}_{u_k,u_l}$, $\forall u_k,u_l \in \mathcal N$, $\forall e_n\in \mathcal K$.
\STATE Initialize $\mathtt C = \varnothing$ and $\Gamma = \varnothing$. 
\STATE \textbf{Solution of Subproblem P2: eRRH-user Transmission}\;
\STATE Design $\mathcal G_{\text{IA-IDNC}}$ according to \sref{IA}.\;
\FOR{\text{each}  $\mathtt S$}
\STATE Calculate
 $\textbf{P}=\{P^*_{e_1},P^*_{e_2},\ \cdots, \, P^*_{e_K}\}$ by solving \eref{eq13}. \;
\STATE Obtain $V_i= \{(R^*_{e_i},P^*_{e_i},s_{e_i}),\cdots, \, (R^*_{e_i},P^*_{e_i},s_{|\mathtt u(\mathtt f^c_{e_i})|}) \}$ ($i=1,\cdots, K$) according to $\mathbf P$.\; \STATE Calculate $w(V_i)$ using \eref{eq10}.\;
\ENDFOR
\STATE $\mathcal G_{\text{IA-IDNC}}(\mathtt C ) \leftarrow \mathcal G_{\text{IA-IDNC}}$.\;
\WHILE{$\mathcal G_{\text{IA-IDNC}}(\mathtt C ) \neq \varnothing$}
\STATE $V^*_i=\arg\max_{V_i\in \mathcal G(\mathtt C )} \{w (V_i)\}$.\; 
\STATE Set $\mathtt C \leftarrow \mathtt C \cup V^*_i$ and $\mathcal G_{\text{IA-IDNC}}(\mathtt C ) \leftarrow \mathcal G_{\text{IA-IDNC}}(V^*_i)$.\;
\ENDWHILE
\STATE \textbf{Solution of Subproblem P3: D2D Transmission}
\STATE Design $\mathcal G_\text{d2d}$ according to \sref{D2D}.\;
\STATE $\mathcal G_{\text{d2d}}(\Gamma) \leftarrow \mathcal G_{\text{d2d}}$.\;
\WHILE{$\mathcal G_{\text{d2d}}(\Gamma) \neq \varnothing$}
\STATE $\forall v \in \mathcal G_{\text{d2d}}(\Gamma)$: calculate $\psi(v)$ and $w(v)$ using \eref{eq12} and \eref{eq15}, respectively.\;
\STATE  $v^*=\arg\max_{v\in \mathcal G_{\text{d2d}}(\Gamma)} \{w (v)\}$.\;  
\STATE Set $\Gamma \leftarrow \Gamma \cup v^*$ and obtain $\mathcal G_{\text{d2d}}(\Gamma)$.\;
\ENDWHILE
\STATE Obtain $\mathtt C$ and $\Gamma $.
\end{algorithmic}
\caption{Proposed Joint Approach} \label{alg:LGS}
\end{algorithm}

\textbf{Example:}  We illustrate in this example how to design the IA-IDNC and D2D conflict graphs of the network presented
in Fig. \ref{fig3}. 
\begin{itemize}
\item In Fig. \ref{fig4n}, we plot the IA-IDNC graph, where each vertex represents a possible NC combination that consists of combined associations in each eRRH. For plotting simplification, we did not include the  vertices  that represents one association (no NC). The connections between vertices (circles) is based on the \textbf{TC} condition that explained in \sref{IA}. There are many possible maximal PCs in $\mathcal G_\text{IA-IDNC}$ that are represented by connected vertices. Each  one represents the potential network-coded scheduling of the eRRHs that minimizes the completion time of users. For example, one possible maximal PC shown in red color in Fig. \ref{fig4n} is $\{s_{R^*P^*53}^1, s_{R^*P^*61}^1,s_{R^*P^*22}^2,s_{R^*P^*34}^2\}$. The five indices $e_1, R^*, P^*, 5, 3$ in the first association represent first eRRH, its transmission rate and power level, scheduled user and its requested file, respectively.

\item To ease the understanding of the D2D conflict graph, we plot it only for the first three users $\{u_1, u_2, u_3\}$ of the network presented
in Fig. \ref{fig3} and irrespective to the possible scheduled users to the eRRHs. The D2D conflict graph is shown in Fig. \ref{fig4}, where the conflict conditions \textbf{CC1} and \textbf{CC2} are represented by solid lines and conditions \textbf{CC3} and \textbf{CC4} are represented by dash lines. By \thref{thm:Mweight1}, one possible maximal IS in this graph is \{$v_{5,4,4}^1, v_{5,5,4}^1, v_{1.5,2,2}^3\}$. The first vertex $v_{5,4,4}^1$  represents the transmitting user $u_1$, its rate $r=5$, the scheduled user $u_4$ and its requested file $f_4$, respectively.
\end{itemize}

 \begin{figure}[t]
      \centering
      \begin{minipage}{0.49\textwidth}
          \centering
         \includegraphics[width=0.85\textwidth]{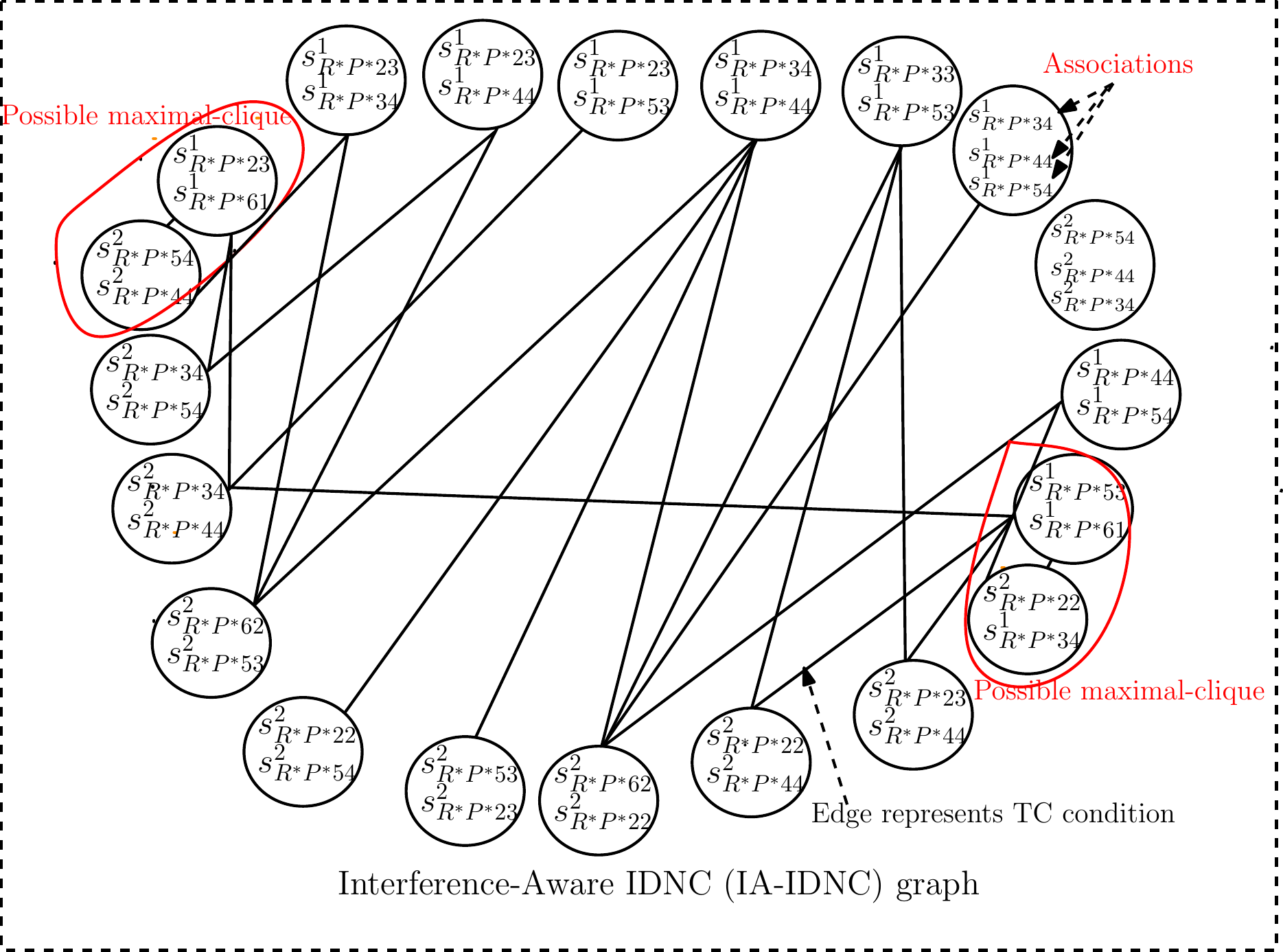} 
          \caption{The IA-IDNC graph. }
          \label{fig4n}
      \end{minipage}\hfill
      \begin{minipage}{0.49\textwidth}
         \includegraphics[width=0.85\textwidth]{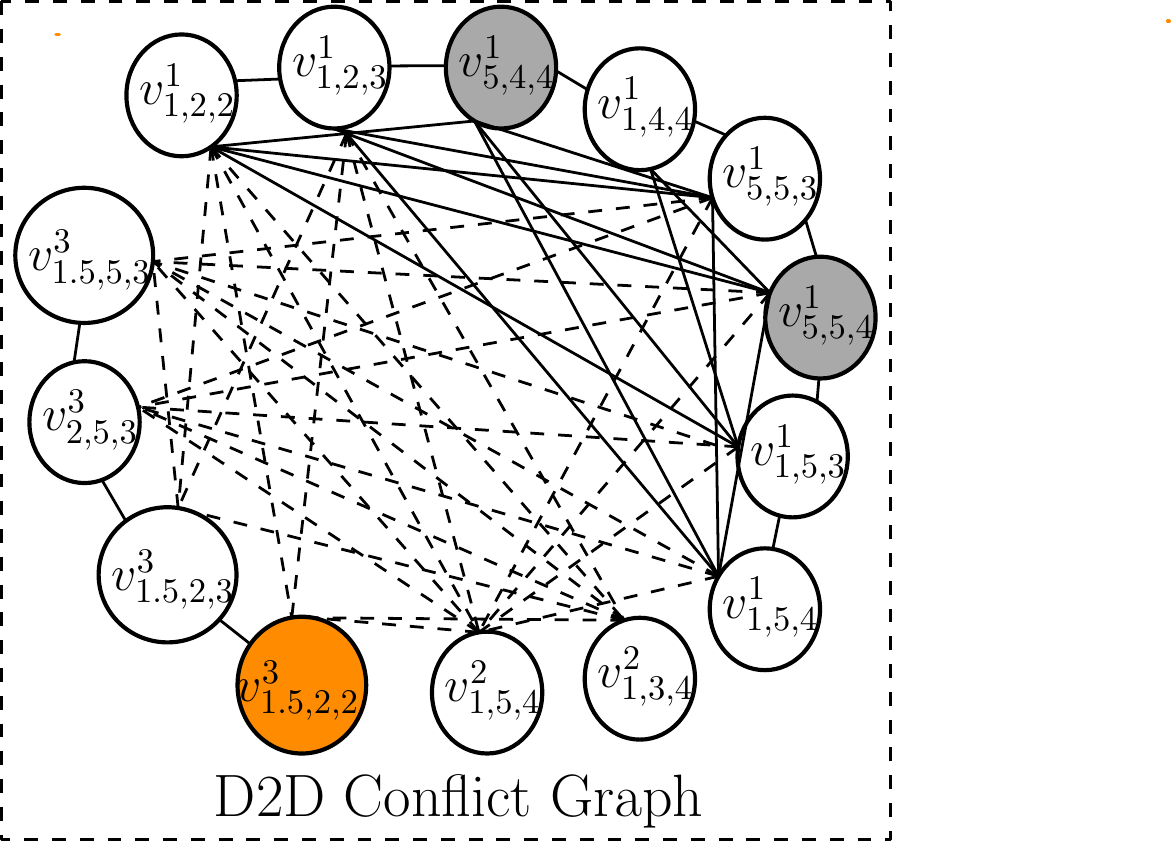} 
          \caption{The D2D conflict graph.}
          \label{fig4}
       \end{minipage}
  \end{figure}


\section{Completion Time Minimization: Coordinated Scheduling Approach} \label{CS}
In this section, we propose  a faster and simpler coordinated scheduling approach. The need for this approach is invoked by the possibly large number of IDNC combinations generated by the joint approach in large network size. In such large networks, we
can utilize this alternative approach to obtain fast
and efficient solution.\ignore{ In particular,  a coordinated scheduling solution is developed to solve the completion time optimization problem without optimizing the power allocation of the eRRHs that considering all NC combinations in the network.}

Let $P_{e_n}$ be a fixed power level of $e_n$-th eRRH, $\forall e_n\in \mathcal K$. The completion time minimization problem at $t$-th transmission slot can be written as a coordinated scheduling problem as follows
\begin{subequations}
\begin{align} \label{eqn:Mschedule}
&\text{P5}: \max_{\substack{
\mathtt f^{c}_{e_n}\subseteq \mathcal P( \mathcal{C}_{e_n})\\ \mathtt u(\mathtt f^c_{e_n})\cap \mathtt u(\mathtt f^c_{e_{n'}}) =\varnothing\\ R_{e_n}\in \mathcal R}} Z_t\\     &{\rm subject ~to\ } 
  \begin{cases}  
   (\mathtt u(\mathtt f^{d2d}_{u_k}), r_{u_k})=\arg \min_{\substack{ \mathcal N_{\text{tra},t}\in \mathcal P(\mathcal{N})\\ \mathtt f^{d2d}_k\in \mathcal P( \mathcal{H}_{u_k})\\ r_{u_k}\in \mathcal{R}_{u_k}}} \left \{\max_{u_l \in \mathcal N_{w,t}} \mathtt T_{u_l,t}\right \},~ \forall u_k\in \mathcal N_{\text{tra},t};   \\
     \mathtt u(\mathtt f^{d2d}_{u_k}) \cap \mathtt u(\mathtt f^{d2d}_{u_{k'}}) =\varnothing, \forall u_k \neq {u_{k'}} \in \mathcal N_{\text{tra},t};\\
     \mathtt f^{d2d}_{u_k}\subseteq \mathcal P( \mathcal{H}_{u_k});
    r_{u_k}\geq R_\text{th}. \label{eq222b}
\end{cases}
       \end{align}
\end{subequations}
The objective function \eref{eq222b} of problem P5 represents the possible completion time minimization that we can obtain from D2D transmissions and \eref{eqn:Mschedule} represents maximizing the number of users left to be scheduled to eRRHs.\ignore{ Clearly, both \eref{eq222b} and \eref{eqn:Mschedule} are network-coded scheduling problems, which are NP hard problems. Therefore, 
this section proposes a low-complexity heuristic solution using a graph theoretic method. }

To solve the problem in P5, we develop a simple and fast approach that first schedules users that have low channel capacities from eRRHs on D2D links, and the remaining unscheduled users (if any) can be scheduled by eRRHs with high transmission rates.  This solution not only minimizes the completion time of users, but also offloads the eRRHs' radio resources. Indeed, after D2D transmissions,  few users left to be scheduled by eRRHs. This approach is summarized in \algref{Alg.2}, which consists of the following two stages at every transmission slot: i) users experience relatively weak 
channels from the eRRHs should be scheduled to the  potential D2D transmitters on D2D links, and
ii) the eRRHs deliver encoded files to a set of users that have not previously scheduled on D2D links.
The  coordinated scheduling approach is described below.

\textbf{{First stage:}} Here, we focus on scheduling a set of users that have low rates from the eRRHs to a set of potential transmitters via D2D links as such we solve problem \eref{eq222b} efficiently.\\
\textit{First step:} Inspired by  \sref{D2D}, we follow the same procedures that construct the new D2D conflict graph $\mathcal G_\text{d2d}(\mathcal V, \mathcal E)$. Thus, we generate a vertex  $v_{r,l,f}^k$ for each  transmitting user $u_k$, a transmission  rate  $r_{u_k} \in \mathcal R^{d2d}_{u_k,u_i}$ and  a missing file  $f_h \in \mathcal H_{u_k}  \cap \mathcal W_{u_l}$ of  a user $u_l \in  \mathcal N_{w}\cap \mathcal Z_{u_k}$. Further, the rate of each transmitting user in each generated vertex should be greater than or equal to $R_{\text{th}}$. Similarly, we generate the vertices  for  $N$ users and then connect them as in \sref{D2D}. \\
\textit{Second step:} We design two-layer weights for each generated vertex in the $\mathcal G_\text{d2d}$ graph, named by \textit{secondary} and \textit{primary} weights. The secondary weight of a vertex $v_{r,l,f}^k$ is defined as $w(v_{r,l,f}^k)=\frac{r_{u_k}}{B}$ that shows a partial contribution of that vertex in terms of reducing the completion time in the network. The primary weigh of a vertex $v_{u_l,f_h}$ is defined as $w(v_{u_i,f_h})=\frac{B}{\min_{e_n\in \mathcal K}R_{e_n,u_l}}, \forall f_h \in \mathcal{C}_{e_n}$. This primary weight  characterizes the users based on their channel capacities from the eRRHs to give them
priority to be scheduled on D2D links. A vertex with high primary weight (low
rate from eRRHs) leads to prolonged file delivery time from eRRHs. Then the corresponding users of such vertices should be scheduled on D2D links with the maximum rate from any possible potential D2D transmitters. As such, the completion time of these users is minimized. \\
\textit{Third step:} We propose to iteratively perform maximum weight search to form the set of D2D transmitters and their scheduled users in the maximal IS  $\Gamma$ as follows. First,
we search for the vertex with the maximum primary weight and find its corresponding  maximum secondary weight. If two or more vertices
have equal weights, we select one vertex randomly. We continue this process
until there are no other available vertices that can
be included in the selected IS. At the end, the final IS $ \Gamma$ consists
of vertices that represent a set of potential D2D transmitters. Each of these D2D transmitters serves users that have low channel capacities from the eRRHs.

\textbf{{Second stage:}} Here, we schedule users that are not scheduled on D2D links to eRRHs using RA-IDNC. In particular, we solve problem \eref{eqn:Mschedule} by maximizing the number of scheduled users to the eRRHs.  First, we construct the coordinated scheduling graph, denoted by $\mathcal G_\text{cord}(\mathcal E, \mathcal V)$, by generating a vertex $V_{e_n,u_i,f_h,R_{e_n}}$ for each $e_n\in \mathcal{K}$, for every file $f_h$ is requested by  user $u_i\in \mathcal{N}_{w,t}\backslash \Gamma$, and for each achievable rate for that user $R_{e_n}\geq r_\text{min}$, where $r_\text{min}$ is the minimum selected transmission rate of any transmitting user in $\Gamma$. The configuration of the set of edges in the scheduling graph is divided into coding (NC and rate edges) and transmission edge. Two vertices $V_{e_n,u_i,f_h,R_{e_n}}$ and $V_{e_n,u_{i^\prime},f_{h^\prime},R_{e_{n^\prime}}}$ representing the same eRRH are adjacent by a conflict edge if they do not satisfy the IDNC and rate conditions in \sref{IA}. Similarly, two vertices $V_{e_n,u_i,f_h,R_{e_n}}$ and $V_{e_{n'},u_{i^\prime},f_{h^\prime},R_{e_{n^\prime}}}$ representing  different eRRHs are adjacent by a conflict transmission  edge if the same user $u_i$ is scheduled to different eRRHs, i.e.,
$u_i=u_{i'}$ and $e_n \neq e_{n'}$. Then, a maximum search process is executed in $\mathcal G_\text{cord}$ to obtain the maximum IS $\mathtt I$ as presented in \algref{Alg.2}.

\begin{algorithm}[t!]
\caption{Proposed Coordinated Scheduling Approach}
\label{Alg.2}
\begin{algorithmic}[1]
\STATE Require $\mathcal N$, $\mathcal F$, $B$, $\mathcal K$,  $\mathcal C_{e_n}$, $\mathcal H_{u_l,0}$, $\mathcal W_{u_l,0}$, $P_{e_n}$,  $Q_{u_k}$, $h^{c}_{e_n,u_l}, h^{d2d}_{u_k,u_l}$, $\forall u_k,u_l \in \mathcal N$, $\forall e_n\in \mathcal K$.
\STATE \textbf{First stage}
\STATE Initialize $\Gamma  = \varnothing$ and  $\mathcal{G}_\text{d2d}(\Gamma) \leftarrow \mathcal{G}_\text{d2d}$.\;
\STATE $\forall v^k_{r,l,h} \in\mathcal{G}_\text{d2d}$: calculate $w(v^k_{r,l,h})=\frac{r_{u_k}}{B}$ and its corresponding  $w(v_{u_l,f_h})=\frac{B}{\min_{e_n\in \mathcal K}R_{e_n,u_l}}$.\;

\WHILE{$\mathcal G_{\text{d2d}}(\Gamma ) \neq c$}
\STATE Choose the maximum primary weight $v_{u_l,f_h}^*$ and finds its corresponding maximum secondary weight $v_{r,l,h}^{k^*}=\arg\max_{v_{r,l,h}^k\in \mathcal G_{\text{d2d}}(\Gamma)} \{w (v_{r,l,h}^{k})\}$.\;
\STATE Set $\Gamma \leftarrow \Gamma \cup v_{r,l,h}^{k^*}$ and $\mathcal{G}_\text{d2d}(\Gamma) \leftarrow \mathcal{G}_\text{d2d}(v_{r,l,h}^{k^*})$.\;
\ENDWHILE

\STATE \textbf{Second stage}
\STATE Design $\mathcal G_\text{cord}$ according to \sref{CS}.\;
\STATE $\forall V_{e_n,u_i,f_h,R_{e_n}}\in \mathcal G_\text{cord}$: calculate $w(V_{e_n,u_i,f_h,R_{e_n}})=\dfrac{R_{e_n}}{B}$.\;
\STATE Obtain the maximum IS $\mathtt I$ as follows.
\STATE Initialize $\mathtt I=\varnothing$.
\FOR{\text{each} $V_{e_n,u_i,f_h,R_{e_n}}\in \mathcal G_\text{cord}$ non-conflicting with $\mathtt I$} 
\STATE Select $V^*_{e_n,u_i,f_h,R_{e_n}}=\arg\max_{V_{e_n,u_i,f_h,R_{e_n}}\in \mathcal G_{\text{cord}}} \{w (V_{e_n,u_i,f_h,R_{e_n}})\}$. \;
\STATE $\mathtt I \leftarrow \mathtt I \cup V^*_{e_n,u_i,f_h,R_{e_n}}.$\;
\ENDFOR
\end{algorithmic}
\end{algorithm}

\ignore{\begin{table}[t!]
       \caption{Numerical parameters} \label{table:parameters}
       \centering
   \begin{tabular}{|p{4.5cm}|p{4.2cm}  |}    \hline
   Bandwidth & 1 MHz \\ \hline
   Channel model  &  SUI-Terrain type B  \\ \hline
   Channel estimation  & Perfect \\ \hline
   Path loss model  & $148 + 40 \log_{10}(\text{dis.[km]})$ \\ \hline
   eRRHs/Users' max power  &   -42.60 dBm/Hz  \\ \hline
   Noise power   & -174 dBm/Hz \\ \hline
   Distribution of users   & Uniform \\ \hline
  Distribution of files   & Random \\ \hline
  eRRH caching ratio $\mu$   & 0.6 \\ \hline
    \end{tabular}
\end{table}}
\section{Numerical Results}\label{N}
This section presents selected simulation results that compare
the completion time performances of our proposed two schemes with baseline algorithms. We consider a
downlink D2D-aided F-RAN system where the eRRHs have fixed locations and users are distributed randomly at every transmission within a hexagonal cell of radius  $900$m. We set the radius of the users' coverage zone $\mathtt R$ to $50$m and the number of eRRHs $K$ to $3$.  We consider the SUI-Terrain type B model in which the channel model of both F-RAN and D2D communications is mostly affected by the location of the users within the cell. Path loss is calculated as $148 + 40 \log_{10}(\text{distance[km]})$. We consider that the channels are perfectly estimated. The noise power and the maximum’ eRRH/user
power are assumed to be $-174$ dBm/Hz and $P_\text{max}=Q=-42.60$ dBm/Hz, respectively. The bandwidth is $1$ MHz and the eRRH caching ratio
$\mu$ is $0.6$. As discussed in  \sref{SMMM}, at the beginning of the D2D-aided F-RAN transmission,  each user already has about $45\%$ and $55\%$ of $F$ files.
To assess the performances of our proposed schemes with different thresholds ($R_{\text{th}1}=0.05$, $R_{\text{th}2}=0.5$, and $R_{\text{th}3}=5$), we simulate various scenarios with different
number of users, number of files, and file sizes.

\ignore{ All  simulated schemes are executed until all files are delivered to all users. The presented average values of the completion time are computed over a certain number of iterations.} For the sake of comparison, our proposed schemes are
compared with the following two baseline NC schemes.
\begin{itemize}
\item \textbf{Random Linear NC (RLNC)}: In RLNC algorithm, each user is scheduled to a single eRRH to which it has the maximum channel capacity. Then, each eRRH encodes all files in its \textit{cache} using random coefficients from Galois field. However,
this algorithm ignores the dynamic transmission rates. To ensure successful delivery of files to users, the selected transmission rate in each eRRH is the minimum channel capacities of all scheduled users.
\item \textbf{Classical IDNC}: For both eRRHs and D2D transmissions, this scheme focuses on network layer optimization, in which the coding decisions depends solely on the file combinations. For successful files' decoding, the transmission rates of both the eRRHs and transmitting users should match the  minimum achievable capacity of all scheduled users. 
\end{itemize}

For completeness of our work, we also compare our proposed schemes with the uncoded schemes.
\begin{itemize}
\item \textbf{Uncoded Unicast}: This scheme schedules only one user to each eRRH from which it receives an uncoded file with its maximum transmission rate. In addition, the untargeted users by the eRRHs is served by implementing uncoded D2D transmissions. 
\item \textbf{Uncoded Broadcast-F-RAN}: The eRRHs broadcast uncoded files sequentially as such all users are served. In this scheme, each eRRH transmits with the lowest transmission rate of all scheduled users.
\item \textbf{Uncoded Broadcast-D2D}: In this scheme, set of transmitting users is selected randomly to broadcast uncoded files from their \textit{Has} sets
that are missing at the largest number of their neighboring users. The transmission rate is selected based on the minimum transmission of all scheduled users.
\end{itemize}

Recently, RA-IDNC scheme is studied in \cite{21n} where all eRRHs use the same transmission rate that corresponds to the minimum transmission rate of all scheduled users. In addition, the unscheduled users by the eRRHs are scheduled from transmitting users over D2D links with the same rate that is used by the eRRHs. Thus, we include the RA-IDNC and compare it with the proposed schemes. 

  \begin{figure}[t!]
      \centering
          \centering
         \includegraphics[width=0.4\textwidth]{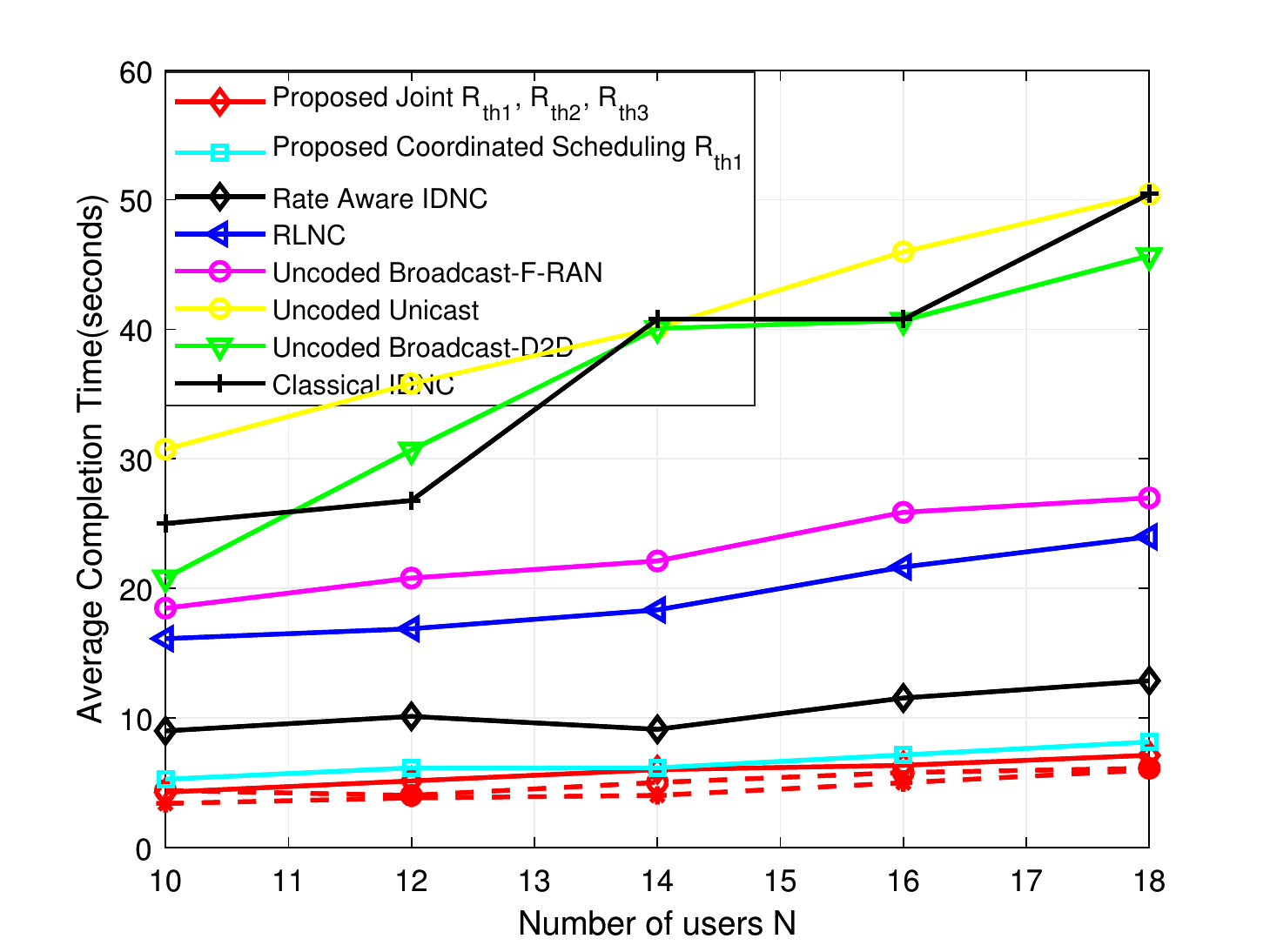} 
          \caption{Average completion time versus the number   of users $N$.}
          \label{fig6}
      \ignore{\begin{minipage}{0.494\textwidth}
          \centering
         \includegraphics[width=0.85\textwidth]{./fig/figfnn.eps} 
          \caption{Average completion time versus the number   of files $F$.}
          \label{fig7}
       \end{minipage}}
  \end{figure} 
  
\ignore{\begin{figure}[t!] 
\centering     
\subfigure[]{\label{fig91}\includegraphics[width=54mm]{./fig/fignn.eps}} \hfill
\subfigure[]{\label{fig9a}\includegraphics[width=54mm]{./fig/fignn.eps}} \hfill
\subfigure[]{\label{fig9b}\includegraphics[width=54mm]{./fig/figfnn.eps}}
\caption{Average completion time versus: (a) Number of users $N$ in small network size; (b) Number of users $N$ in large network size; (c) Number of files $F$.} 
\label{fig6}
\end{figure}}

In Fig. \ref{fig6}, we depict the  average completion time versus the number of users $N$. We consider a D2D-aided F-RAN model with a frame of $15$ files and a file size of $1$ Mbits. From this figure, it can be seen that the  proposed schemes offer improved performance in terms of completion time reduction as compared to the other schemes. This improved performance is due to the joint 
and coordinated schemes that (i) judiciously schedule users, adopt the transmission rate of each eRRH and optimize the transmission power of each eRRH, and (ii) select potential users for transmitting coded files over D2D links. In particular, the uncoded unicast suffers from targeting few users that have relative strong channel qualities. As a result, a higher number of transmissions, at least $(N*F)/(K+|\mathcal N_\text{tra}|)$ transmissions, is needed for frame delivery completion, and it leads to a high completion time. Uncoded broadcast schemes suffer from serving all users at the cost of adopting the transmission rate of all eRRHs and transmitting users with the minimum transmission rate of the served users. Furthermore, uncoded broadcast D2D scheme offers a poor completion time performance as all transmitting users do not benefit from the transmission, i.e., they cannot transmit and receive at the same time. RLNC is a rate-less scheme that targets all users by sending encoded file with the lowest rate of all users. On the other hand, RA-IDNC scheme offers an improved performance compared to uncoded, RLNC, and classical IDNC schemes as mentioned in \cite{21n}.
This is because the coding decisions in RA-IDNC scheme not only depends on the file combinations, but also
on the channel qualities of the scheduled users. This effectively balances between the number of scheduled users and the transmission rate of eRRHs/transmitting users. However, selecting one transmission rate (the minimum rate) for all eRRHs and transmitting users degrades the completion time performance of the RA-IDNC scheme. This is a clear limitation of the RA-IDNC scheme in [23], as it does
not fully exploit the typical variable channel qualities of the different eRRHs/transmitting users to their scheduled users. Our proposed joint and coordinated schemes fully utilize the eRRHs and transmitting users' resources to choose their own transmission rates, XOR combinations, and scheduled users. Consequently, a better performance of our proposed schemes compared to the RA-IDNC scheme is achieved. Moreover, the joint scheme optimizes the employed rates using power control on each eRRH. Thus, it works better than our proposed coordinated scheme. Note that the completion time performances of the classical IDNC and uncoded broadcast D2D schemes are of orders $10^5$ and $10^3$, respectively. Thus, we omit them from all the remaining figures.

\ignore{To evaluate the performance of our proposed
schemes in a large D2D-aided F-RAN, we plot in Fig. \ref{fig6}-(b) the average completion time versus a large number of users $N$ with $K=5$ eRRHs, $F = 30$ files and file size of $1$ Mbits. Again, for
the above-mentioned reasons in Fig. \ref{fig6}-(a), our proposed schemes outperform the other schemes. Intuitively, RA-IDNC scheme exhibits a significant degradation with large number of eRRHs. This can be explained by the fact that RA-IDNC always selects the minimum transmission rates of the eRRHs/transmitting users.}

In Fig. \ref{fig7}, we show the average completion time versus the number of files $F$. The simulated D2D-aided F-RAN system composed of $15$ users and file size of $1$ Mbits. Again, for
the above-mentioned reasons in Fig. \ref{fig6}, our proposed schemes outperform the other schemes. It can be observed from the figure that increasing the frame size leads to an increasing in the completion time of all schemes. The opportunities of mixing files using NC in the RA-IDNC and proposed schemes are limited with few files. Therefore, all NC schemes have roughly similar performances. As the number of files increases, the  increase in the completion time with our proposed schemes is low.  This is in accordance
with our results in \thref{thm:Mweight} and \thref{thm:Mweight1}, where it is shown that our proposed schemes judiciously allow each eRRH and each transmitting user to decide on a set of files to be XORed. As such, they are beneficial to a significant set of users that have relatively good channel qualities. Even though uncoded broadcast and RLNC schemes completes file transmissions in fewer transmissions ($F$ transmissions) than our proposed schemes, each of their transmission durations is
longer than a single transmission of the proposed schemes. Thanks to their optimized higher rate/power and users' transmission that result in less completion time. 

 \begin{figure}[t!]
      \centering
      \begin{minipage}{0.494\textwidth}
                \centering
               \includegraphics[width=0.8\textwidth]{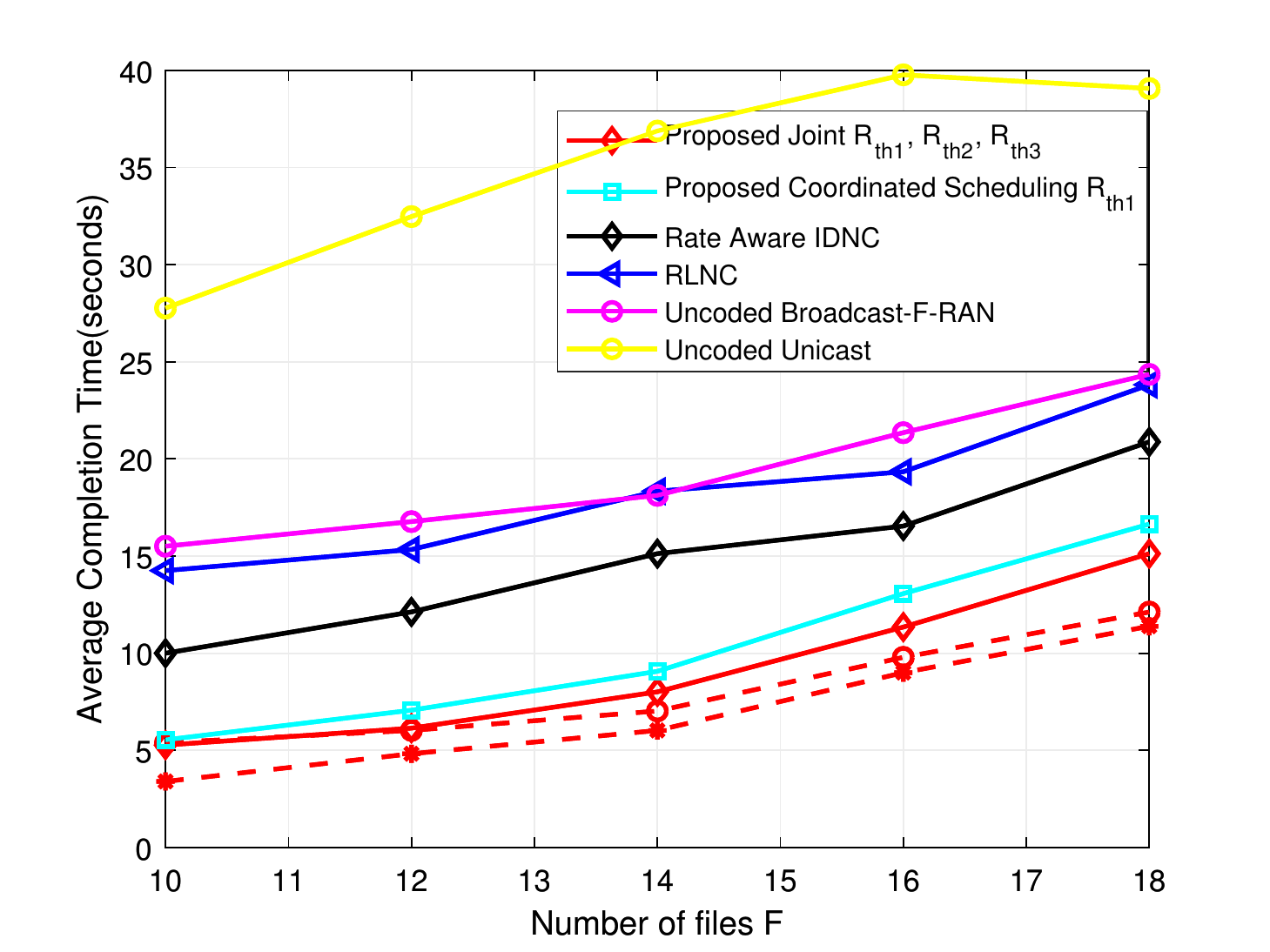} 
                \caption{Average completion time versus the number   of files $F$.}
                \label{fig7}
             \end{minipage}\hfill
      \begin{minipage}{0.494\textwidth}
          \centering
         \includegraphics[width=0.8\textwidth]{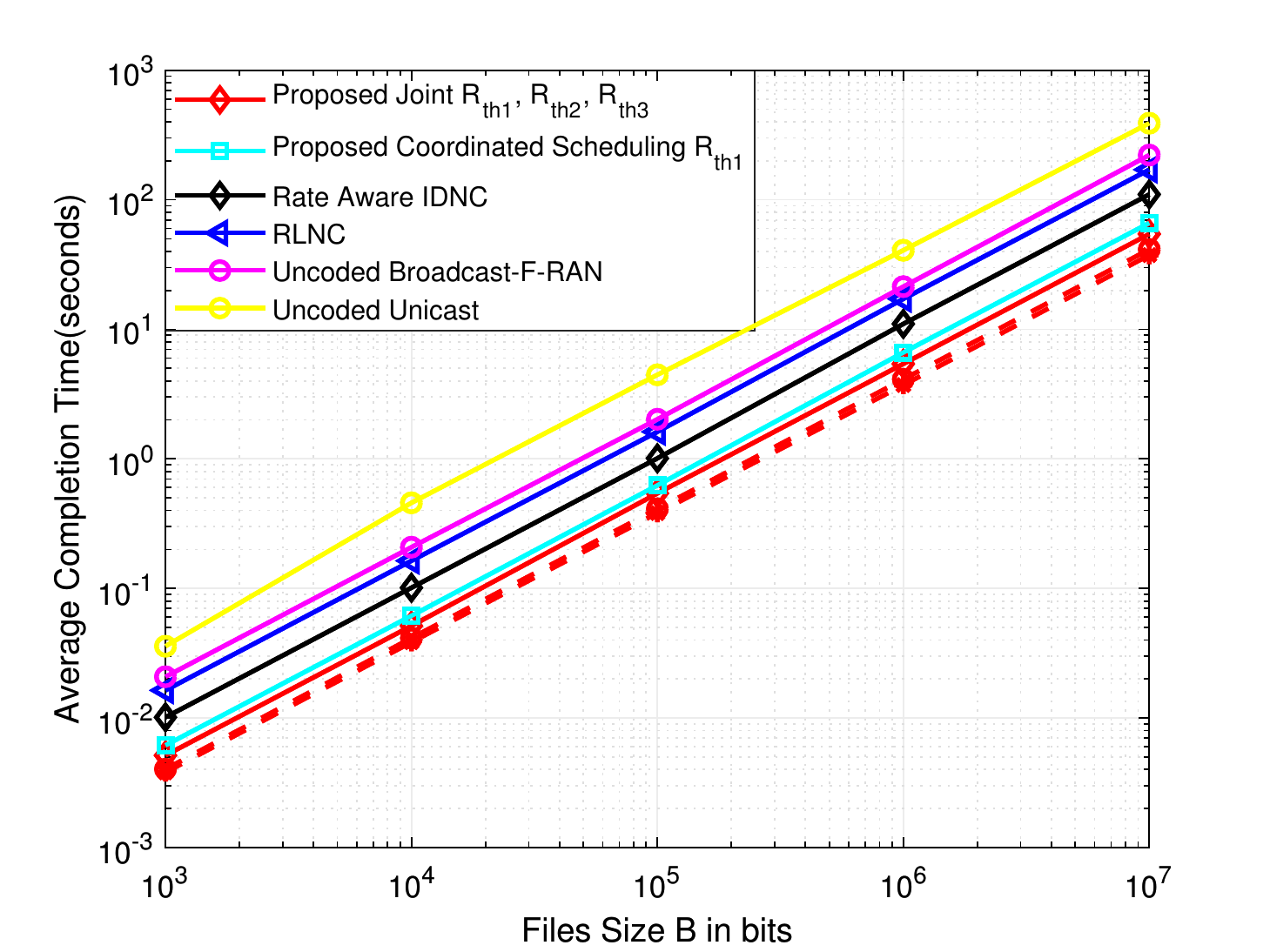} 
          \caption{Average completion time versus file size $B$.}
          \label{fig8}
      \end{minipage}\hfill
      \ignore{\begin{minipage}{0.494\textwidth}
          \centering
         \includegraphics[width=0.85\textwidth]{./fig/figfnn.eps} 
          \caption{Average completion time versus percentage of users' side information.}
          \label{fig9}
       \end{minipage}}
  \end{figure} 
  
\ignore{\begin{figure}
\centering     
\subfigure[]{\label{fig11}\includegraphics[width=54mm]{./fig/figbn.eps}} \hfill
\subfigure[]{\label{fig12}\includegraphics[width=54mm]{./fig/figbn.eps}} \hfill
\subfigure[]{\label{fig13}\includegraphics[width=54mm]{./fig/figbn.eps}}
\caption{Average completion time versus: (a) Percentage of the side information; (b) File size $B$; (c) Cell radius $\mathtt R$.} 
\label{fig7}
\end{figure}}

In Fig. \ref{fig8}, we illustrate the impact of increasing the file size $B$ on the average completion time. In this figure, we simulate the D2D-aided F-RAN system composed of $12$ users and $15$ files. We can observe that the performances of all schemes increase linearly with the file size. This is in accordance
with the completion time expression in \corref{cor:LBcompletion}, where it was emphasized that  $\mathtt T_o$ increases linearly with $B$. From physical-layer view, as $B$ increases, more bits are needed for delivering files. Thus, time delay is increased to receive files from eRRHs/transmitting users. 

Finally, some insights from our presented numerical results are given as follows. First, it is advantageous to serve many users with NC files as in the classical IDNC algorithm, but selecting the minimum transmission rate of all scheduled users degrades its performance. Thus, this scheme is impractical from physical-layer perspective. Second, although the uncoded unicast scheme uses the maximum transmission rate of each user, it needs a large number of transmissions for completion. Thus, its completion time performance is degraded.   Third, RA-IDNC scheme overcomes the limitations of the aforementioned schemes but suffers from selecting the lowest-rate of the fixed-power eRRHs in the system. This limitation further degrades the completion time performance of the RA-IDNC scheme in large network sizes with large number of eRRHs. This due to the fact that
RA-IDNC always selects the lowest-rate of
all eRRHs. Conversely, our transmission framework is more practically  relevant as it enables different transmission rates from
different eRRHs/transmitting users and optimizes the employed rates using power control on each eRRH.

\ignore{In Fig. \ref{fig7}-(c), we study the effect of increasing the coverage zones of users on the completion times of all simulated schemes. The considered parameters are $15$ users, $18$ files, file size is $1$ Mbits.  It can be seen that $\mathtt R$ has no effect on the average completion time of uncoded broadcast-F-RAN and RLNC schemes. As}
\ignore{Finally, Fig. 9 shows the effect of increasing the number of eRRHs on the completion time of the proposed  and RA-IDNC (the one has good performance compared to all other schemes) schemes. For comparison convenient, we only simulate the proposed and RA-IDNC schemes. Due to the minimum rate selection and fixed power level of all eRRHs in RA-IDNC scheme, we expect that its completion time performance suffers. The increase of the completion time is more with RA-IDNC scheme when the number of eRRHs increases. It is intuitively to note that selecting a minimum rate of an increasing set of eRRHs always gives a minimum rate. However, our proposed schemes give more flexibility to the eRRHs and transmitting users to choose their transmission rates, which has a potential effect on reducing the completion time.}

\section{Conclusion}\label{C}
We have developed a framework that exploits the cached contents at
eRRHs, their transmission rates/powers, and previously received contents by different users to deliver the
requesting contents to users with a minimum completion time in D2D-aided F-RAN. Towards this target, we have first formulated an optimization problem that seeks to minimize the completion time of users and decomposed it into two subproblems. Specifically, the first subproblem is to minimize the transmission durations of eRRHs. To solve it, we have designed an interference-aware IDNC graph that considers network-coded scheduling and power allocation for each eRRH. Based on this solution, the second subproblem maximizes the number of unscheduled users to eRRHs via D2D links. Then, we have introduced a new D2D conflict graph that
achieves an effective solution to the second subproblem based
on greedy search method. The aforementioned graph-based
solutions of the corresponding subproblems are referred to a joint approach. For the high implementation complexity of the joint approach in large networks, we have developed an alternative and efficient coordinated scheme that has relatively low implementation complexity. Simulation results have shown that our proposed schemes can effectively minimize the frame delivery time as compared to conventional schemes.


\end{document}